%% file: main.tex
\documentclass{article}
\pdfoutput=1

\input{macros}

\begin{document}

\title{A Distributed Frank-Wolfe Framework for Learning Low-Rank Matrices with the Trace Norm
}


\author{Wenjie Zheng         \and
        Aur\'elien Bellet \and 
        Patrick Gallinari
}

%
%
%

\maketitle

\input{abstract}

\input{intro}

\input{background}

\input{dfw}

\input{related}

\input{exp}

\input{conclu}

\section*{Acknowledgements}
This work was partially supported by ANR Pamela (grant ANR-16-CE23-0016-01) and by a grant from CPER Nord-Pas de Calais/FEDER DATA Advanced data science and technologies 2015-2020.
The first author would like to thank Ludovic Denoyer, Hubert Naacke, Mohamed-Amine Baazizi, and the engineers of LIP6 for their help during the deployment of the cluster.

\bibliographystyle{spbasic}      
\bibliography{main}   

\begin{appendices}
\input{appendix}

\end{appendices}

\end{document}

%% file: macros.tex
\usepackage[utf8]{inputenc}
\usepackage{amsmath}
\usepackage{amsfonts}
\usepackage{amssymb}
\usepackage{amsthm}
\usepackage{graphicx}
\usepackage{color}
\usepackage{natbib}
\usepackage{xspace}
\usepackage{url}
\usepackage[title]{appendix}
\usepackage[left=3cm,right=3cm,top=3cm,bottom=3cm]{geometry}

\usepackage{algorithm}
\usepackage{algpseudocode}

\DeclareMathOperator*{\argmin}{arg\,min}
\DeclareMathOperator{\diam}{diam}

\newcommand{\dfw}{\textsc{DFW-Trace}\xspace}
\newcommand{\naive}{\textsc{Naive-DFW}\xspace}
\newcommand{\sva}{\textsc{SVA}\xspace}

\theoremstyle{definition}
\newtheorem{definition}{Definition}
\theoremstyle{plain}
\newtheorem{theorem}{Theorem}
\newtheorem{lemma}[theorem]{Lemma}
\newtheorem{remark}{Remark}

%% file: abstract.tex

\begin{abstract}
We consider the problem of learning a high-dimensional but low-rank matrix from a large-scale dataset distributed over several machines, where low-rankness is enforced by a convex trace norm constraint. 
We propose \dfw, a distributed Frank-Wolfe algorithm which leverages the low-rank structure of its updates to achieve efficiency in time, memory and communication usage. The step at the heart of \dfw is solved approximately using a distributed version of the power method. We provide a theoretical analysis of the convergence of \dfw, showing that we can ensure sublinear convergence in expectation to an optimal solution with few power iterations per epoch. We implement \dfw in the Apache Spark distributed programming framework and validate the usefulness of our approach on synthetic and real data, including the ImageNet dataset with high-dimensional features extracted from a deep neural network.
\end{abstract}

%% file: intro.tex

\section{Introduction}
\label{sec:intro}

Learning low-rank matrices is a problem of great importance in machine learning, statistics and computer vision. Since rank minimization is known to be NP-hard, a principled approach consists in solving a convex relaxation of the problem where the rank is replaced by the trace norm (also known as the nuclear norm) of the matrix.
This strategy is supported by a range of theoretical results showing that when the ground truth matrix is truly low-rank, one can recover it exactly (or accurately) from limited samples and under mild conditions \citep[see][]{bach2008consistency,candes2009exact,candes2010power,recht2011simpler,gross2010quantum,gross2011recovering,koltchinskii2011nuclear,bhojanapalli2016global}.
Trace norm minimization has led to many successful applications, among which collaborative filtering and recommender systems \citep{koren2009matrix}, multi-task learning \citep{Argyriou2008a,pong2010trace}, multi-class and multi-label classification \citep{goldberg2010transduction,cabral2011matrix,harchaoui2012large}, robust PCA \citep{cabral2013unifying}, phase retrieval \citep{candes2015phase} and video denoising \citep{ji2010robust}.

We consider the following generic formulation of the problem:
\begin{equation}
\min_{W\in\mathbb{R}^{d\times m}}F\left(W\right)=\sum_{i=1}^nf_i(W)\quad\text{s.t. }\|W\| _{*}\le\mu,
\label{eq:tn-opt}
\end{equation}
where the $f_i$'s are differentiable with Lipschitz-continuous gradient, $\|W\| _{*}=\sum_k \sigma_k(W)$ is the trace norm of $W$ (the sum of its singular values), and $\mu>0$ is a regularization parameter (typically tuned by cross-validation). In a machine learning context, an important class of problems considers $f_i(W)$ to be a loss value which is small (resp. large) when $W$ fits well (resp. poorly) the $i$-th data point (see Section~\ref{sec:app} for concrete examples).\footnote{More general cases can be addressed, such as pairwise loss functions $f_{i,j}$ corresponding to pairs of data points.} In this work, we focus on the large-scale scenario where the quantities involved in \eqref{eq:tn-opt} are large: typically, the matrix dimensions $d$ and $m$ are both in the thousands or above, and the number of functions (data points) $n$ is in the millions or more. 

Various approaches have been proposed to solve the trace norm minimization problem \eqref{eq:tn-opt}.\footnote{Some methods consider an equivalent formulation where the trace norm appears as a penalization term in the objective function rather than as a constraint.} One can rely on reformulations as semi-definite programs and use out-of-the-shelf solvers such as SDPT3 \citep{toh1999sdpt3} or SeDuMi \citep{sturm1999using}, but this does not scale beyond small-size problems. To overcome this limitation, first-order methods like Singular Value Thresholding \citep{cai2010singular}, Fixed Point Continuation algorithms \citep{ma2011fixed} and more generally projected/proximal gradient algorithms \citep{Parikh2013a} have been proposed. These approaches have two important drawbacks preventing their use when the matrix dimensions $d$ and $m$ are both very large: they require to compute a costly (approximate) SVD at each iteration, and their memory complexity is $O(dm)$. In this context, Frank-Wolfe (also known as conditional gradient) algorithms \citep{frank1956algorithm} provide a significant reduction in computational and memory complexity: they only need to compute the leading eigenvector at each iteration, and they maintain compact low-rank iterates throughout the optimization \citep{hazan2008sparse,jaggi2010simple,jaggi2013revisiting,Harchaoui2015a}. 
However, as all first-order algorithms, Frank-Wolfe requires to compute the gradient of the objective function at each iteration, which requires a full pass over the dataset and becomes a bottleneck when $n$ is large.

The goal of this paper is to propose a \emph{distributed version of the Frank-Wolfe algorithm} in order to alleviate the cost of gradient computation when solving problem \eqref{eq:tn-opt}. We focus on the Bulk Synchronous Parallel (BSP) model with a master node connected to a set of slaves (workers), each of the workers having access to a subset of the $f_i$'s (typically corresponding to a subset of training points).
Our contributions are three-fold. First, we propose \dfw, a Frank-Wolfe algorithm relying on a distributed power method to approximately compute the leading eigenvector with communication cost of $O(d+m)$ per pass over the dataset (\emph{epoch}). This dramatically improves upon the $O(dm)$ cost incurred by a naive distributed approach. Second, we prove the sublinear convergence of \dfw  to an optimal solution in expectation, quantifying the number of power iterations needed at each epoch. This result guarantees that \dfw can find low-rank matrices with small approximation error using few power iterations per epoch.
Lastly, we provide a modular implementation of our approach in the Apache Spark programming framework \citep{Zaharia:2010:SCC:1863103.1863113} which can be readily deployed on commodity and commercial clusters. We evaluate the practical performance of \dfw by applying it to multi-task regression and multi-class classification tasks on synthetic and real-world datasets, including the ImageNet database \citep{Deng2009a} with high-dimensional features generated by a deep neural network. The results confirm that \dfw has fast convergence and outperforms competing methods. While distributed FW algorithms have been proposed for other classes of problems \citep{Bellet2015b,moharrerdistributing,wang2016parallel}, to the best of our knowledge our work is the first to propose, analyze and experiment with a distributed Frank-Wolfe algorithm designed specifically for trace norm minimization.

The rest of this paper is organized as follows. Section~\ref{sec:background} introduces some background on the (centralized) Frank-Wolfe algorithm and its specialization to trace norm minimization, and reviews some applications. After presenting some baseline approaches for the distributed setting, Section~\ref{sec:dfw} describes our algorithm \dfw and its convergence analysis, as well as some implementation details. Section~\ref{sec:related} discusses some related work, and Section~\ref{sec:exp} presents the experimental results.

%% file: background.tex

\section{Background}
\label{sec:background}

We review the centralized Frank-Wolfe algorithm in Section~\ref{sec:fw} and its specialization to trace norm minimization in Section~\ref{sec:trace}. We then present some applications to multi-task learning and multi-class classification in Section~\ref{sec:app}.

\subsection{Frank-Wolfe Algorithm}
\label{sec:fw}

\begin{algorithm}[t]
\caption{Centralized Frank-Wolfe algorithm to solve \eqref{eq:fw-gen}}
\label{alg:fw}
\begin{algorithmic}[0]
\State \textbf{Input:} Initial point $W^0 \in \mathcal{D}$, number of iterations $T$
\For{$t = 0,\dots,T-1$}
    \State $S^* \leftarrow \arg\min_{S\in\mathcal{D}} \langle S, \nabla F(W^t) \rangle$
    \Comment solve linear subproblem
    \State $\gamma^t \leftarrow \frac{2}{t+2}$ (or determined by line search)
    \Comment step size
    \State $W^{t+1} \leftarrow (1-\gamma^t) W^t + \gamma^t S^*$
    \Comment update
\EndFor
\State \textbf{Output:} $W^T$
\end{algorithmic}
\end{algorithm}

The original Frank-Wolfe (FW) algorithm dates back from the 1950s and was originally designed for quadratic programming \citep{frank1956algorithm}. The scope of the algorithm was then extended to sparse greedy approximation \citep{clarkson2010coresets} and semi-definite programming \citep{hazan2008sparse}. Recently, \citet{jaggi2013revisiting} generalized the algorithm further to tackle the following generic problem:
\begin{equation}
\min_{W\in\mathcal{D}}F\left(W\right),
\label{eq:fw-gen}
\end{equation}
where $F$ is convex and continuously differentiable, and the feasible domain $\mathcal{D}$ is a compact convex subset of some Hilbert space with inner product $\langle \cdot,\cdot\rangle$.  

Algorithm~\ref{alg:fw} shows the generic formulation of the FW algorithm applied to \eqref{eq:fw-gen}. At each iteration $t$, the algorithm finds the feasible point $S^*\in\mathcal{D}$ which minimizes the linearization of $F$ at the current iterate $W^t$. The next iterate $W^{t+1}$ is then obtained by a convex combination of $W^t$ and $S^*$, with a relative weight given by the step size $\gamma^t$. By convexity of $\mathcal{D}$, this ensures that $W^{t+1}$ is feasible. The algorithm converges in $O(1/t)$, as shown by the following result from \citet{jaggi2013revisiting}.

\begin{theorem}[\citealp{jaggi2013revisiting}]
\label{thm:fw}
Let $C_F$ be the curvature constant of $F$.\footnote{This constant is bounded above by $L\diam(\mathcal{D})^2$, where $L$ is the Lipschitz constant of the gradient of $F$ \citep[see][]{jaggi2013revisiting}.} For each $t\geq 1$, the iterate $W^t\in\mathcal{D}$ generated by Algorithm~\ref{alg:fw} satisfies $F(W^t) - F(W^*) \leq \frac{2C_F}{t+2}$, where $W^*\in\mathcal{D}$ is an optimal solution to \eqref{eq:fw-gen}.
\end{theorem}

\begin{remark}
There exist several variants of the FW algorithm, for which faster rates can sometimes be derived under additional assumptions. We refer to \cite{jaggi2013revisiting}, and \cite{Lacoste-Julien2015a} for details.
\end{remark}

From the algorithmic point of view, the main step in Algorithm~\ref{alg:fw} is to solve the linear subproblem over the domain $\mathcal{D}$. By the linearity of the subproblem, a solution always lies at an extremal point of $\mathcal{D}$, hence FW can be seen as a greedy algorithm whose iterates are convex combinations of extremal points (adding a new one at each iteration). When these extremal points have some specific structure (e.g., sparsity, low-rankness), the iterates inherit this structure and the linear subproblem can sometimes be solved very efficiently. This is the case for the trace norm constraint, our focus in this paper.

\subsection{Specialization to Trace Norm Minimization}
\label{sec:trace}

The FW algorithm applied to the trace norm minimization problem \eqref{eq:tn-opt} must solve the following subproblem:
\begin{equation}
S^*\in\argmin_{\|S\|_{*}\le\mu}\langle S,\nabla F(W^{t})\rangle,
\label{eq:tn-sub}
\end{equation}
where $W^t\in\mathbb{R}^{d\times m}$ is the iterate at time $t$ and $S\in\mathbb{R}^{d\times m}$.
The trace norm ball is the convex hull of the rank-1 matrices, so there must exist a rank-1 solution to \eqref{eq:tn-sub}. This solution can be shown to be equal to $-\mu u_1v_1^\top$, where $u_1$ and $v_1$ are the unit left and right top singular vectors of the gradient matrix $\nabla F(W^{t})$ \citep{jaggi2013revisiting}. Finding the top singular vectors of a matrix is much more efficient than computing the full SVD. This gives FW a significant computational advantage over projected and proximal gradient descent approaches when the matrix dimensions are large. Furthermore, assuming that $W^0$ is initialized to the zero matrix, $W^{t}$ can be stored in a compact form as a convex combination of $t$ rank-1 matrices, which requires $O(t(d+m))$ memory instead of $O(dm)$ to store a full rank matrix. As implied by Theorem~\ref{thm:fw}, FW is thus guaranteed to find a rank-$t$ whose approximation error is $O(1/t)$ for any $t\geq 1$. In practice, when the ground truth matrix is indeed low-rank, FW can typically recover a very accurate solution after $t\ll \min(d, m)$ steps.

We note that in the special case where the matrix $W$ is square ($d=m$) and constrained to be symmetric, the gradient $\nabla F(W^{t})$ can always be written as a symmetric matrix, and the solution to the linear subproblem has a simpler representation based on the leading eigenvector of the gradient, see \citet{jaggi2013revisiting}.

\subsection{Applications}
\label{sec:app}

We describe here two tasks where trace norm minimization has been successfully applied, which we will use to evaluate our method in Section~\ref{sec:exp}.\\

\noindent\textbf{Multi-task least square regression.} This is an instance of multi-task learning \citep{Caruana1997a}, where one aims to jointly learn $m$ related tasks. Formally, let $X\in\mathbb{R}^{n\times d}$ be the feature matrix ($n$ training points in $d$-dimensional space) and $Y\in\mathbb{R}^{n\times m}$ be the response matrix (each column corresponding to a task). The objective function aims to minimize the residuals of all tasks simultaneously:
\begin{equation}
\label{eq:mls}
F(W)=\frac{1}{2}\left\Vert XW-Y\right\Vert _{F}^{2}=\frac{1}{2}\sum_{i=1}^n\sum_{j=1}^m(x_{i}^{T}w_{j}-y_{ij})^{2},
\end{equation}
where $\|\cdot\|_F$ is the Frobenius norm.
Using a trace norm constraint on $W$ allows to couple the tasks together by making the task predictors share a common subspace, which is a standard approach to multi-task learning \citep[see e.g.,][]{Argyriou2008a,pong2010trace}.\\

\noindent\textbf{Multinomial logistic regression.} Consider a classification problem with $m$ classes. Let $X\in\mathbb{R}^{n\times d}$ be the feature matrix and $y\in\{1,\dots, m\}^n$ the label vector. Multinomial logistic regression minimizes the negative log-likelihood function:
\begin{equation}
F(W)=\sum_{i}\log\Big(1+\sum_{l\ne y_{i}}\exp(w_{l}^{T}x_{i}-w_{y_{i}}^{T}x_{i})\Big)=\sum_{i}\Big(-w_{y_{i}}^{T}x_{i}+\log\sum_{l}\exp(w_{l}^{T}x_{i})\Big). \label{eq:mlr}
\end{equation}
The motivation for using the trace norm is that multi-class problems with a large number of categories usually exhibit low-rank embeddings of the classes \citep[see][]{Amit2007a,harchaoui2012large}.\\


%% file: dfw.tex

\section{Distributed Frank-Wolfe for Trace Norm Minimization}
\label{sec:dfw}

We now consider a distributed master/slave architecture with $N$ slaves (workers). The master node is connected to all workers and acts mainly as an aggregator, while most of the computation is done on the workers. The individual functions $f_1,\dots,f_n$ in the objective \eqref{eq:tn-opt} are partitioned across workers, so that all workers can collectively compute all functions but each worker can only compute its own subset. Recall that in a typical machine learning scenario, each function $f_i$ corresponds to the loss function computed on the $i$-th data point (as in the examples of Section~\ref{sec:app}). We will thus often refer to these functions as data points.
Formally, let $I_j\subseteq\{1,\dots,n\}$ be the set of indices assigned to worker $j$, where $I_1\cup\dots\cup I_N=\{1,\dots,n\}$ and $I_1\cap\dots\cap I_N=\emptyset$. We denote by $F_j=\sum_{i\in I_j} f_i$ the local function (dataset) associated with each worker $j$, and by $n_j=|I_j|$ the size of this local dataset.

We follow the Bulk Synchronous Parallel (BSP) computational model: each iteration (\emph{epoch}) alternates between parallel computation at the workers and communication with the master (the latter serves as a synchronization barrier).

\subsection{Baseline Strategies}
\label{sec:baseline}

Before presenting our algorithm, we first introduce two baseline distributed FW strategies (each with their own merits and drawbacks).\\

\noindent\textbf{Naive DFW.}
One can immediately see a naive way of running the centralized Frank-Wolfe algorithm (Algorithm~\ref{alg:fw}) in the distributed setting. Starting from a common initial point $W^0$, each worker $j$ computes at each iteration $t$ its local gradient $\nabla F_j(W^t)$ and sends it to the master. The master then aggregates the messages to produce the full gradient $\nabla F(W^t)=\sum_{j=1}^N F_j(W^t)$, solves the linear subproblem by computing the leading right/left singular vectors of $\nabla F(W^t)$ and sends the solution back to the workers, who can form the next iterate $W^{t+1}$. \naive exactly mimics the behavior of the centralized FW algorithm, but induces a communication cost of $O(Ndm)$ per epoch as in many applications (such as those presented in Section~\ref{sec:app}) the local gradients are dense matrices. In the large-scale setting where the matrix dimensions $d$ and $m$ are both large, this cost dramatically limits the efficiency of the algorithm.\\

\noindent\textbf{Singular Vector Averaging.} A possible strategy to avoid this high communication cost is to ask each worker $j$ to send to the master the rank-1 solution to the local version of the subproblem \eqref{eq:tn-sub}, in which they use their local gradient $\nabla F_j(W^t)$ as an estimate of the full gradient $\nabla F(W^t)$. This reduces the communication to a much more affordable cost of $O(N(d+m))$.
Note that averaging the rank-1 solutions would typically lead to a rank-$N$ update, which breaks the useful rank-1 property of FW and is undesirable when $N$ is large. Instead, the master averages the singular vectors (weighted proportionally to $n_j$), resolving the sign ambiguity by setting the largest entry of each singular vector to be positive and using appropriate normalization, as mentioned for instance in \citet{Bro2008a}. We refer to this strategy as Singular Vector Averaging (\sva). \sva is a reasonable heuristic when the individual functions are partitioned across nodes uniformly at random: in this case the local gradients can be seen as unbiased estimates of the full gradient. However the singular vector estimate itself is biased (averaging between workers only reduces its variance), and for $n$ fixed this bias increases with the matrix dimensions $d$ and $m$ but also with the number of workers $N$ (which is not a desirable property in the distributed setting). It is also expected to perform badly on arbitrary (non-uniform) partitions of functions across workers. Clearly, one cannot hope to establish strong convergence guarantees for \sva.

\subsection{Proposed Approach}

We now describe our proposed approach, referred to as \dfw. We will see that \dfw achieves roughly the small communication cost of \sva while enjoying a similar convergence rate as \naive (and hence centralized FW).\\

\begin{algorithm}[t]
\caption{Our distributed algorithm \dfw to solve \eqref{eq:tn-opt}}
\label{alg:dfw}
\begin{algorithmic}[1]
\State \textbf{Input:} Initial point $W^0 \in \mathcal{D}$, number of iterations $T$
\For{$t = 0,\dots,T-1$}
    \State \textbf{Each worker $j$:} $\nabla F_{j}(W^t) \leftarrow \sum_{i\in I_j}\nabla f_{i}(W^t)$
    \State \textbf{All workers:} draw the same $v_{0}\in\mathbb{R}^m$ uniformly on unit sphere
    \For{$k = 0,\dots,K(t)-1$}\Comment{distributed power method}\label{line:pmstart}
        \State \textbf{Each worker $j$:} send $u_{k+1,j}\leftarrow \nabla F_{j}(W^t)v_{k}$ to master
        \State \textbf{Master:} broadcast $u_{k+1}\leftarrow(\sum_{j=1}^{N}u_{k+1,j}) / \|\sum_{j=1}^{N}u_{k+1,j}\|$ 
        \State \textbf{Each worker $j$:} send $v_{k+1,j}\leftarrow \nabla F_{j}(W^t)^\top u_{k+1}$ to master
        \State \textbf{Master:} broadcast $v_{k+1}\leftarrow(\sum_{j=1}^{N}v_{k+1,j})/\|\sum_{j=1}^{N}v_{k+1,j}\|$
    \EndFor\label{line:pmend}
    \State $\gamma^t \leftarrow \frac{2}{t+2}$ (or determined by line search)
    \Comment step size
    \State \textbf{Each worker $j$:} $W^{t+1} \leftarrow (1-\gamma^t) W^t - \gamma^t \mu u_{K(t)}v_{K(t)}^\top $
    \Comment update
\EndFor
\State \textbf{Output:} $W^T$
\end{algorithmic}
\end{algorithm}


\noindent\textbf{Algorithm.}
The main idea of \dfw (Algorithm~\ref{alg:dfw}) is to solve the linear subproblem of FW approximately using a distributed version of the power method applied to the matrix $\nabla F(W^t)^\top F(W^t)$. At each outer iteration (epoch) $t$, the workers first generate a common random vector drawn uniformly on the unit sphere.\footnote{This can be done without communication: for instance, the workers can agree on a common random seed before running the algorithm.} Then, for $K(t)$ iterations, the algorithm alternates between the workers computing matrix-vector products and the master aggregating the results. At the end of this procedure, workers hold the same approximate versions of the left and right singular vectors of $\nabla F(W^t)$ and use them to generate the next iterate $W^{t+1}$.

The communication cost of \dfw per epoch is $O(NK(t)(d+m))$ (see Table~\ref{tab:com} for a comparison with baselines). It is clear that as $K(t)\rightarrow\infty$, \dfw computes the exact solution to the linear subproblems and hence has the same convergence guarantees as centralized FW. However, we would like to set $K(t)\ll \min(d,m)$ to provide a significant improvement over the $O(Ndm)$ cost of the naive distributed algorithm. The purpose of our analysis below is to show how to set $K(t)$ to preserve the convergence rate of the centralized algorithm.

\begin{table}[t]
\caption{Communication cost per epoch of the various algorithms. $K(t)$ is the number of power iterations used by \dfw.}
\label{tab:com}
\centering
\begin{tabular}{ccc}
\hline\noalign{\smallskip}
 Algorithm & Communication cost & \# communication rounds \\ 
\noalign{\smallskip}\hline\noalign{\smallskip}
Naive FW & $Ndm$ & 1 \\ 
Singular Vector Averaging & $N(d+m)$ & 1 \\ 
\dfw & $2NK(t)(d+m)$ & $2K(t)$ \\ 
\noalign{\smallskip}\hline
\end{tabular}
\end{table}

\begin{remark}[Other network topologies]
Since any connected graph can be logically represented as a star graph by choosing a center, our method virtually works on any network (though it may incur additional communication). Depending on the topology, special care can be taken to reduce the communication overhead. An interesting case is the rooted tree network: we can adopt a hierarchical aggregation scheme which has the same communication cost of $O(NK(t)(d+m))$ as the star network but scales better to many workers by allowing parallel aggregations.\footnote{In Apache Spark, this is implemented in \texttt{treeReduce} and \texttt{treeAggregate}.} For a general graph with $M$ edges, $O(MK(t)(d+m))$ communication is enough to broadcast the values to all workers so they can perform the aggregation locally.
\end{remark}

\noindent\textbf{Analysis.}
We will establish that for some appropriate choices of $K(t)$, \dfw achieves sublinear convergence in expectation, as defined below.

\begin{definition}
\label{def:conv}
Let $\delta\geq 0$ be an accuracy parameter. We say that \dfw \emph{converges sublinearly in expectation} if for each $t\ge 1$, its iterate $W^t$ satisfies
\begin{equation}
\label{eq:con-exp}
\mathbb{E}[F(W^t)] - F(W^*) \le \tfrac{2C_F}{t+2}(1+\delta),
\end{equation}
where $C_F$ is the curvature constant of $F$.
\end{definition}
We have the following result.

\begin{theorem}[Convergence]
\label{thm:conv-exp}
Let $F$ be a convex, differentiable function with curvature $C_F$ and Lipschitz constant $L$ w.r.t. the trace norm. For any accuracy parameter $\delta\geq 0$, the following properties hold for \dfw (Algorithm~\ref{alg:dfw}):
\begin{enumerate}
\item If $m\geq 8$ and for any $t\geq 0$, $K(t) \geq 1+\lceil\frac{\mu L (t+2) \ln m}{\delta C_F}\rceil$, then \dfw converges sublinearly in expectation.
\item For any $t\geq 0$, let $\sigma_1^t$, $\sigma_2^t$ be the largest and the second largest singular values of $\nabla F(W^t)$ and assume that $\sigma_1^t$ has multiplicity 1 and there exists a constant $\beta$ such that $\tfrac{\sigma_2^t}{\sigma_1^t} < \beta <1 $. If $K(t) \geq \max ( \lceil \frac{\ln (\delta C_F) - \ln [m \mu L (t+2)]}{2 \ln \beta} \rceil + 1, \widetilde{K})$ where $\widetilde{K}$ is a large enough constant, \dfw converges sublinearly in expectation.
\end{enumerate}
\end{theorem}
\begin{proof}[Sketch]
We briefly outline the main ingredients (see Appendix~\ref{sec:proof1} for the detailed proof). We first show that if the linear subproblem is approximately solved in expectation (to sufficient accuracy), then the FW algorithm converges sublinearly in expectation. Relying on results on the convergence of the power method \citep{kuczynski1992estimating} and on the Lipschitzness of $F$, we then derive the above results on the number of power iterations $K(t)$ needed to ensure sufficient accuracy under different assumptions.
\end{proof}

Theorem~\ref{thm:conv-exp} characterizes the number of power iterations $K(t)$ at each epoch $t$ which is sufficient to guarantee that \dfw converges sublinearly in expectation to an optimal solution. Note that there are two regimes. The first part of the theorem establishes that if $K(t)$ scales linearly in $t$, the expected output of \dfw after $t$ epochs is a rank-$t$ matrix with $O(1/t)$ approximation error (as in centralized FW, see Theorem~\ref{thm:fw}).
In the large-scale setting of interest, this implies that a good low-rank approximation can be achieved by running the algorithm for $t\ll\min(d, m)$ iterations, and with reasonable communication cost since $K(t) = O(t)$.
Remarkably, this result holds without any assumption about the spectral structure of the gradient matrices. On the other hand, in the regime where the gradient matrices are ``well-behaved'' (in the sense that the ratio between their two largest singular values is bounded away from $1$), the second part of the theorem shows that a much lower number of power iterations $K(t)=O(\log t)$ is sufficient to ensure the sublinear convergence in expectation. In Section~\ref{sec:exp}, we will see experimentally on several datasets that this is indeed sufficient in practice to achieve convergence.
We conclude this part with a few remarks mentioning some additional results, for which we omit the details due to the lack of space.

\begin{remark}[Convergence in probability]
We can also establish the sublinear convergence of \dfw \emph{in probability} (which is stronger than convergence in expectation). The results are analogous to Theorem~\ref{thm:fw} but require $K(t)$ to be quadratic in $t$ for the first case, and linear in $t$ for the second case.
\end{remark}

\begin{remark}[Constant number of power iterations]
If we take the number of power iterations to be constant across epochs (i.e., $K(t) = K$ for all $t$), \dfw converges in expectation to a neighborhood of the optimal solution whose size decreases with $K$. We can establish this by combining results on the approximation error of the power method with Theorem~5.1 in \citet{Freund2016a}.
\end{remark}

\subsection{Implementation}
\label{sec:imple}

Our algorithm \dfw (Algorithm~\ref{alg:dfw}) can be implemented as a sequence of map-reduce steps \citep{dean2008mapreduce}. This allows the computation to be massively parallelized across the set of workers, while allowing a simple implementation and fast deployment on commodity and commercial clusters via existing distributed programming frameworks \citep{dean2008mapreduce,Zaharia:2010:SCC:1863103.1863113}.
Nonetheless, some special care is needed if one wants to get an efficient implementation. In particular, it is key to leverage the fundamental property of FW algorithms that the updates are rank-1. This structural property implies that it is much more efficient to compute the gradient in a recursive manner, rather than from scratch using the current parameters.
We use a notion of \emph{sufficient information} to denote the local quantities (maintained by each worker) that are sufficient to compute the updates. This includes the local gradient (for the reason outlined above), and sometimes some quantities precomputed from the local dataset.
Depending on the objective function and the relative size of the problem parameters $n$, $m$, $d$ and $N$, the memory and/or time complexity may be improved by storing (some of) the sufficient information in low-rank form. We refer the reader to Appendix~\ref{sec:imp-app} for a concrete application of these ideas to the tasks of multi-task least square regression and multinomial logistic regression used in our experiments.

Based on the above principles, we developed an open-source Python implementation of \dfw using the Apache Spark framework \citep{Zaharia:2010:SCC:1863103.1863113}.\footnote{\url{https://github.com/WenjieZ/distributed-frank-wolfe}} The package also implements the baseline strategies of Section~\ref{sec:baseline}, and currently uses dense representations. The code is modular and separates generic from task-specific components. In particular, the generic \dfw algorithm is implemented in PySpark (Spark's Python API) in a task-agnostic fashion. On the other hand, specific tasks (objective function, gradient, etc) are implemented separately in pure Python code. This allows users to easily extend the package by adding their own tasks of interest without requiring Spark knowledge.
Specifically, the task interface should implement several methods: \verb+stats+ (to initialize the sufficient information), \verb+update+ (to update the sufficient information), and optionally \verb+linesearch+ (to use linesearch instead of default step size) and \verb+loss+ (to compute the value of the objective function).
In the current version, we provide such interface for multi-task least square regression and multinomial logistic regression.

%% file: related.tex

\section{Related Work}
\label{sec:related}

There has been a recent surge of interest for the Frank-Wolfe algorithm and its variants in the machine learning community. The renewed popularity of this classic algorithm, introduced by \citet{frank1956algorithm}, can be largely attributed to the work of \citet{clarkson2010coresets} and more recently \citet{jaggi2013revisiting}. They generalized its scope and showed that its strong convergence guarantees, efficient greedy updates and sparse iterates are valuable to tackle high-dimensional machine learning problems involving sparsity-inducing (non-smooth) regularization such as the $L_1$ norm and the trace norm. Subsequent work has extended the convergence results, for instance proving faster rates under some additional assumptions \citep[see][]{Lacoste-Julien2015a,Garber2015a,Freund2016a}.

As first-order methods, FW algorithms rely on gradients. In machine learning, computing the gradient of the objective typically requires a full pass over the dataset. To alleviate this computational cost on large datasets, some distributed versions of FW algorithms have recently been proposed for various problems. \citet{Bellet2015b} introduced a communication-efficient distributed FW algorithm for a class of problems under $L_1$ norm and simplex constraints, and provided an MPI-based implementation. \citet{Tran2015a} extend the algorithm to the Stale Synchronous Parallel (SSP) model. \citet{moharrerdistributing} further generalized the class of problems which can be considered (still under $L_1$/simplex constraints) and proposed an efficient and modular implementation in Apache Spark (similar to what we propose in the present work for trace norm problems).
\citet{wang2016parallel} proposed a parallel and distributed version of the Block-Coordinate Frank-Wolfe algorithm \citep{Lacoste-Julien2013b} for problems with block-separable constraints. All these methods are designed for specific problem classes and do not apply to trace norm minimization.
For general problems (including trace norm minimization), \citet{wai2017decentralized} introduced a decentralized FW algorithm in which workers communicate over a network graph without master node. The communication steps involve local averages of iterates and gradients between neighboring workers. In the master/slave distributed setting we consider, their algorithm essentially reduces to the naive distributed FW described in Section~\ref{sec:baseline} and hence suffers from the high communication cost induced by transmitting gradients. In contrast to the above approaches, our work proposes a communication-efficient distributed FW algorithm for trace norm minimization.
In parallel to our work, \citet{wai2017fast} recently introduced an extension of the approach by \citet{wai2017decentralized} specialized to trace norm minimization. Their method also relies on the power method, but their work is complementary to ours. In particular, they consider the gossip communication setting and derive convergence results of different nature. Their focus is on the theoretical analysis: they do not discuss an efficient implementation of their algorithm, and the experiments are merely illustrative (simulated runs in a centralized environment). 


Another direction to scale up FW algorithms to large datasets is to consider stochastic variants, where the gradient is replaced by an unbiased estimate computed on a mini-batch of data points \citep{Hazan2012a,Lan2016a,Hazan2016a}. The price to pay is a slower theoretical convergence rate, and in practice some instability and convergence issues have been observed \citep[see e.g.,][]{liu2017approximate}. The experimental results of \citet{moharrerdistributing} show that current stochastic FW approaches do not match the performance of their distributed counterparts. Despite these limitations, this line of work is largely complementary to ours: when the number of workers $N$ is small compared to the training set size $n$, each worker could compute an estimate of its local gradient to further reduce the computational cost. We leave this for future work.

We conclude this section by mentioning that other kinds of distributed algorithms have been proposed for special cases of our general problem \eqref{eq:tn-opt}. In particular, for the matrix completion problem, \citet{mackey2011divide} proposed a divide-and-conquer strategy, splitting the input matrix into submatrices, solving each subproblem in parallel with an existing matrix completion algorithm, and then combining the results.

%% file: exp.tex

\section{Experiments}
\label{sec:exp}

In this section, we validate the proposed approach through experiments on two tasks: multi-task least square regression and multinomial logistic regression (see Section~\ref{sec:app}). We use both synthetic and real-world datasets.

\subsection{Experimental Setup}

\noindent\textbf{Environment.}
We run our Spark implementation described in Section~\ref{sec:imple} on a cluster with 5 identical machines, with Spark 1.6 deployed in standalone mode. One machine serves as the driver (master) and the other four as executors (workers). Each machine has 2 Intel Xeon E5645 2.40GHz CPUs, each with 6 physical cores. Each physical core has 2 threads. Therefore, we have 96 logical cores available as workers. The Spark cluster is configured to use all 96 logical cores unless otherwise stated. Each machine has 64GB RAM: our Spark deployment is configured to use 60GB, hence the executors use 240GB in total. The network card has a speed of 1Gb/s. The BLAS version does not enable multi-threading.\\

\noindent\textbf{Datasets.}
For multi-task least square, we experiment on synthetic data generated as follows. The ground truth $W$ has rank 10 and trace norm equal to 1 (we thus set $\mu=1$ in the experiments). This is obtained by multiplying two arbitrary orthogonal matrices and a sparse diagonal matrix. $X$ is generated randomly, with each coefficient following a Gaussian distribution, and we set $Y = XW$. We generate two versions of the dataset: a low-dimensional dataset ($n=10^5$ samples, $d=300$ features and $m=300$ tasks) and a higher dimensional one ($n=10^5$, $d=1,000$ and $m=1,000$).
For multinomial logistic regression, we use a synthetic and a real dataset. The synthetic dataset has $n=10^5$ samples, $p=1,000$ features and $m=1,000$ classes. The generation of $W$ and $X$ is the same as above, with the label vector $y$ set to the one yielding the highest score for each point. The test set has $10^5$ samples. Our real-world dataset is ImageNet from ILSVRC2012 challenge \citep{Deng2009a,ILSVRC15}, which has $n=1,281,167$ training images in $m=1,000$ classes. We use the learned features of dimension $p=2048$ extracted from the deep neural network ResNet50 \citep{he2016deep} provided by Keras.\footnote{\url{https://github.com/fchollet/keras}} The validation set of the competition ($50,000$ images) serves as the test set.
\\

\noindent\textbf{Compared methods.}
We compare the following algorithms: \naive, \sva (the baselines described in Section~\ref{sec:baseline}) and three variants of our algorithm, \dfw-1, \dfw-2 and \dfw-log (resp. using 1, 2 and $O(\log t)$ power iterations at step $t$). We have also experimented with \dfw with $K(t)=O(t)$, but observed empirically that far fewer power iterations are sufficient in practice to ensure good convergence. We have also used SVA as warm start to the power iterations within \dfw, which marginally improves the performance of \dfw. We do not show these variants on the figures for clarity.

\subsection{Results}

\begin{figure}[t]
\centering
\includegraphics[width=.45\textwidth]{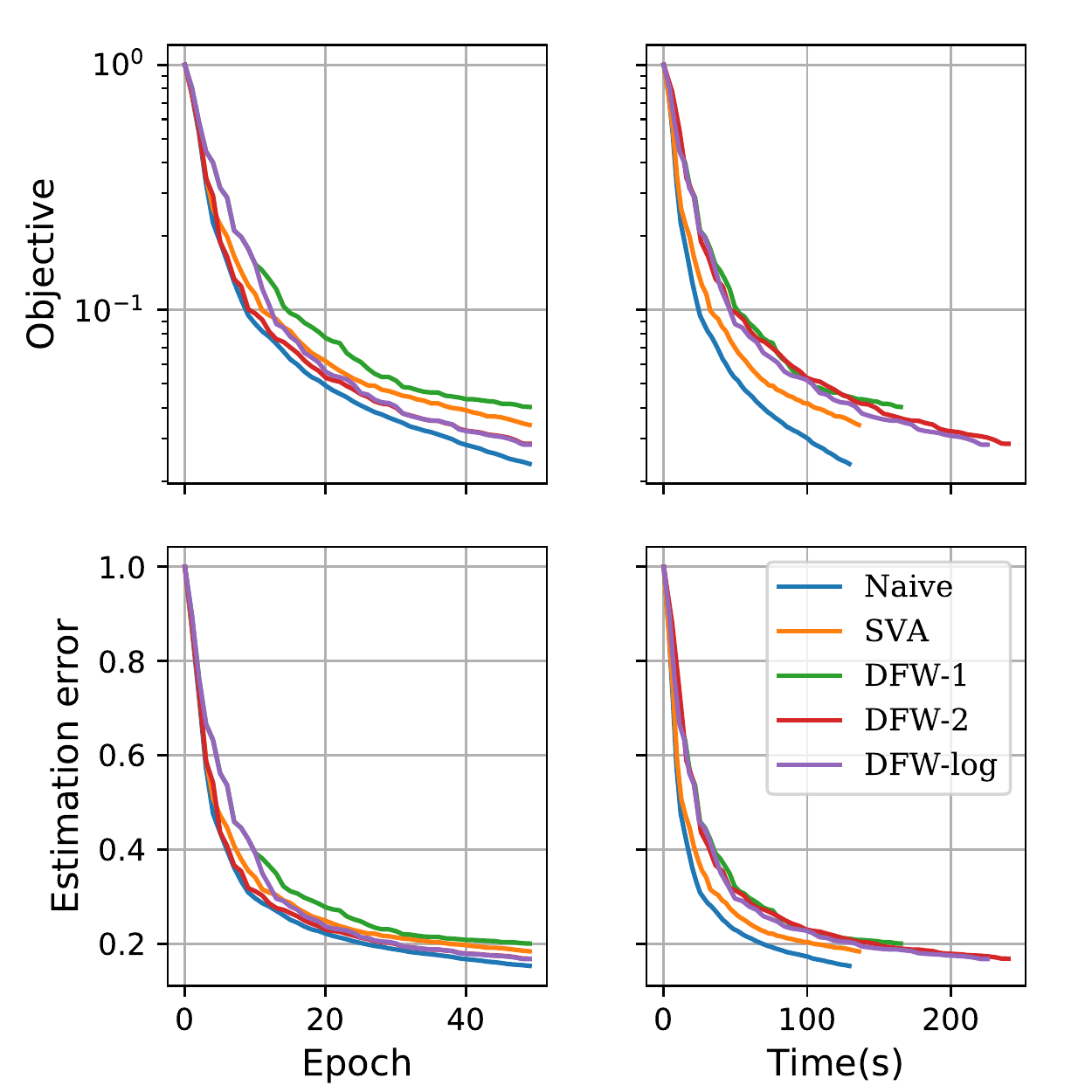} 
\includegraphics[width=.45\textwidth]{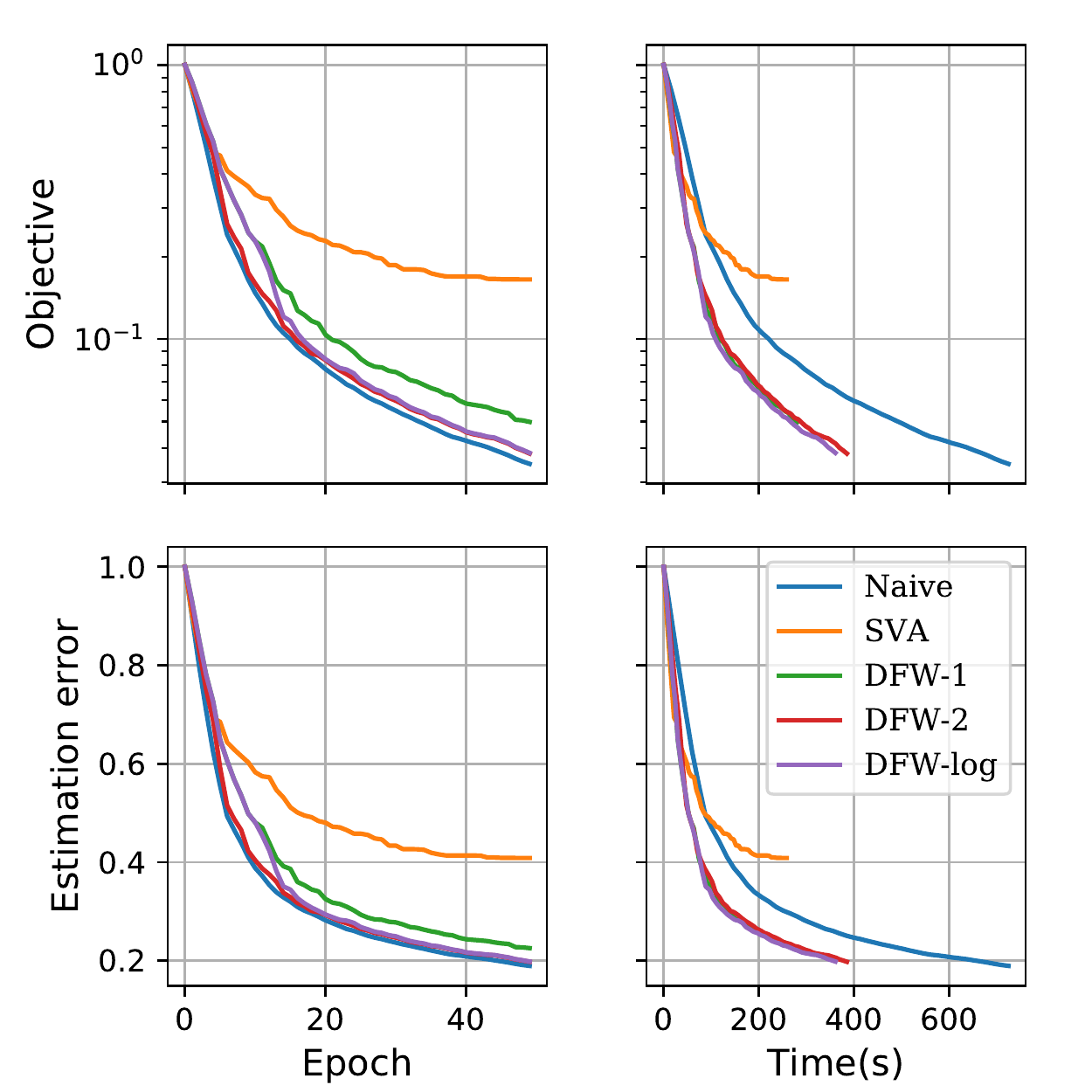} 
\caption{Results for multi-task least square regression. Left: low-dimensional dataset ($n=10^5$, $d=300$ and $m=300$). Right: higher-dimensional dataset ($n=10^5$, $d=1,000$ and $m=1,000$).}
\label{fig:mls}
\end{figure}

\noindent\textbf{Multi-task least square.} For this task, we simply set the number of power iterations of \dfw-log to $K(t)=\lfloor1+\log(t)\rfloor$. All algorithms use line search. Figure~\ref{fig:mls} shows the results for all methods on the low and high-dimensional versions of the dataset. The performance is shown with respect to the number of epochs and runtime, and for two metrics: the value of the objective function and the estimation error (relative Frobenius distance between the current $W$ and the ground truth). On this dataset, the estimation error behaves similarly as the objective function.
As expected, \naive performs the best with respect to the number of epochs as it computes the exact solution to the linear subproblem. On the low-dimensional dataset (left panel), it also provides the fastest decrease in objective/error. \sva also performs well on this dataset. However, when the dimension grows (right panel) the accuracy of \sva drops dramatically and \naive becomes much slower due to the increased communication cost. This confirms that these baselines do not scale well with the matrix dimensions. On the other hand, all variants of \dfw perform much better than the baselines on the higher-dimensional dataset. This gap is expected to widen as the matrix dimensions increase. Remarkably, only 2 power iterations are sufficient to closely match the reduction in objective function achieved by the exact solution on this task. One can see the influence of the number of power iterations on the progress per epoch (notice for instance the clear break at iteration $10$ when \dfw-log switches from $1$ to $2$ power iterations), but this has a cost in terms of runtime. Overall, all variants of \dfw reduce the objective/error at roughly the same speed. On a smaller scale version of the dataset, we verified that the gradients are well-behaved in the sense of Theorem~\ref{thm:conv-exp}: the average ratio between the two largest singular values over 100 epochs was found to be $0.86$.
\\

\begin{figure}[t]
\centering
\includegraphics[width=.325\textwidth]{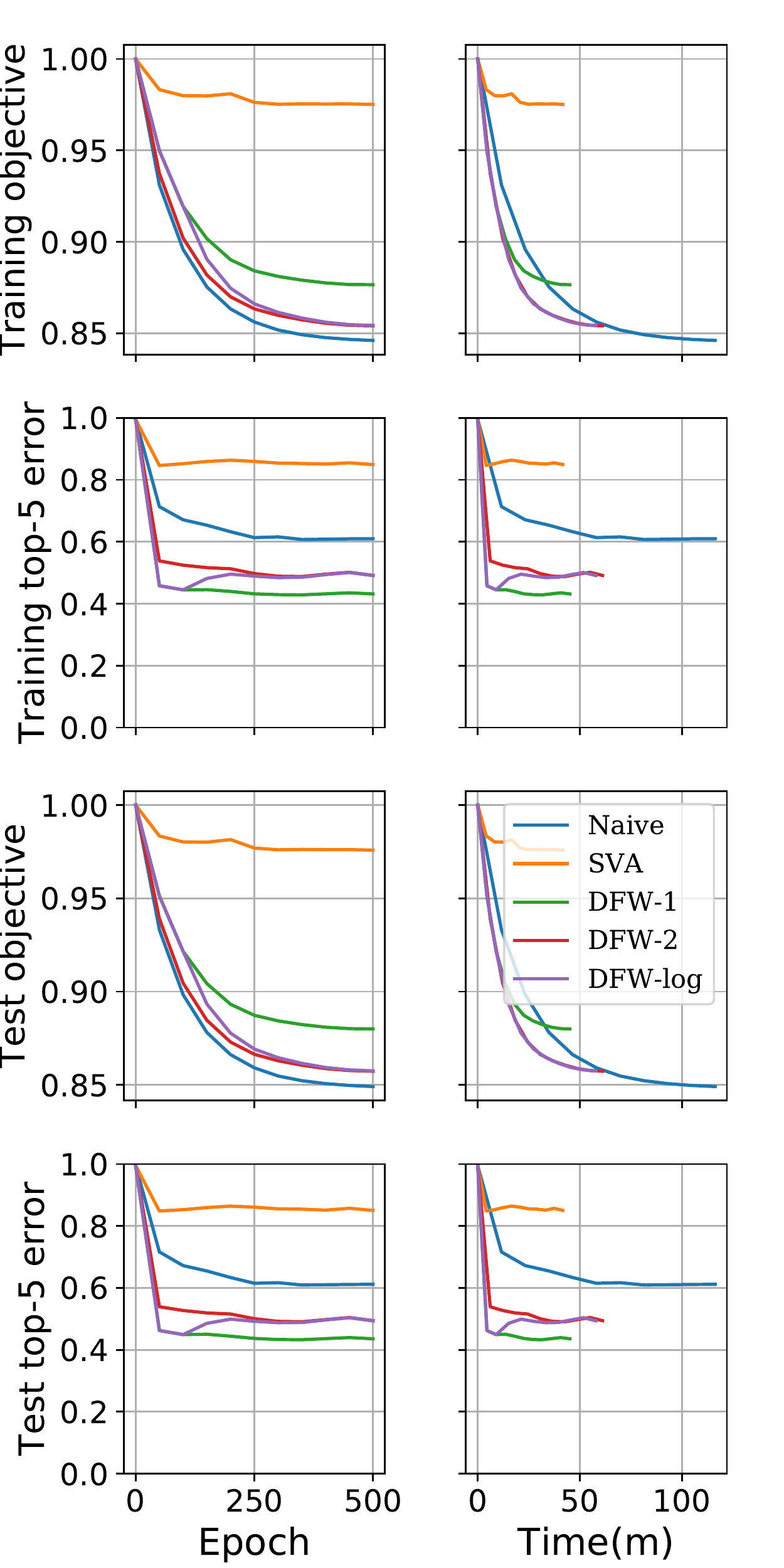} 
\includegraphics[width=.325\textwidth]{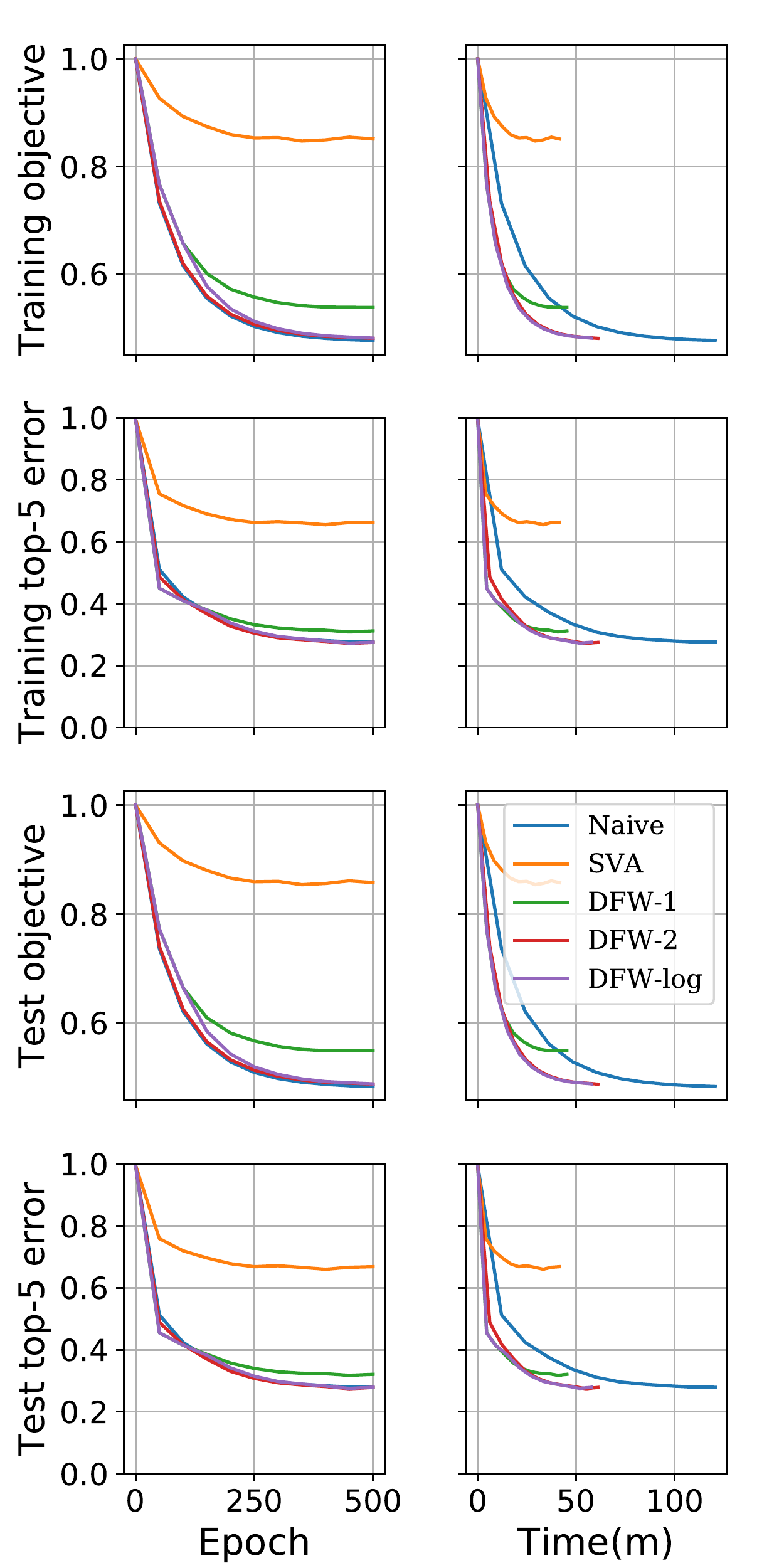} 
\includegraphics[width=.325\textwidth]{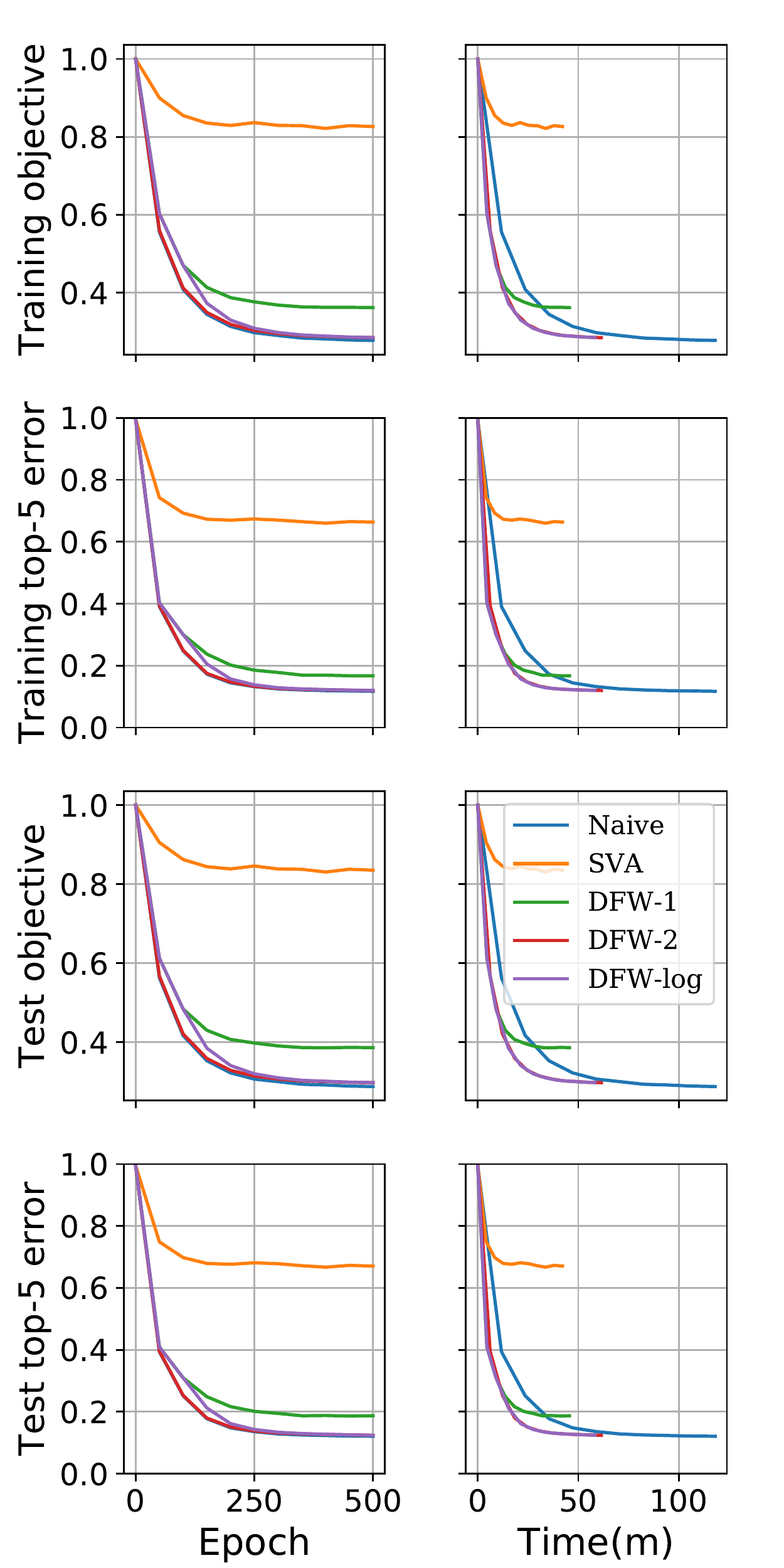} 
\caption{Results for multinomial logistic regression (synthetic data) for several values of $\mu$. Left: $\mu = 10$. Middle: $\mu = 50$. Right: $\mu = 100$. The error stands for the top-5 misclassification rate.}
\label{fig:mlr}
\end{figure}

\begin{figure}[t]
\centering
\includegraphics[width=.47\textwidth]{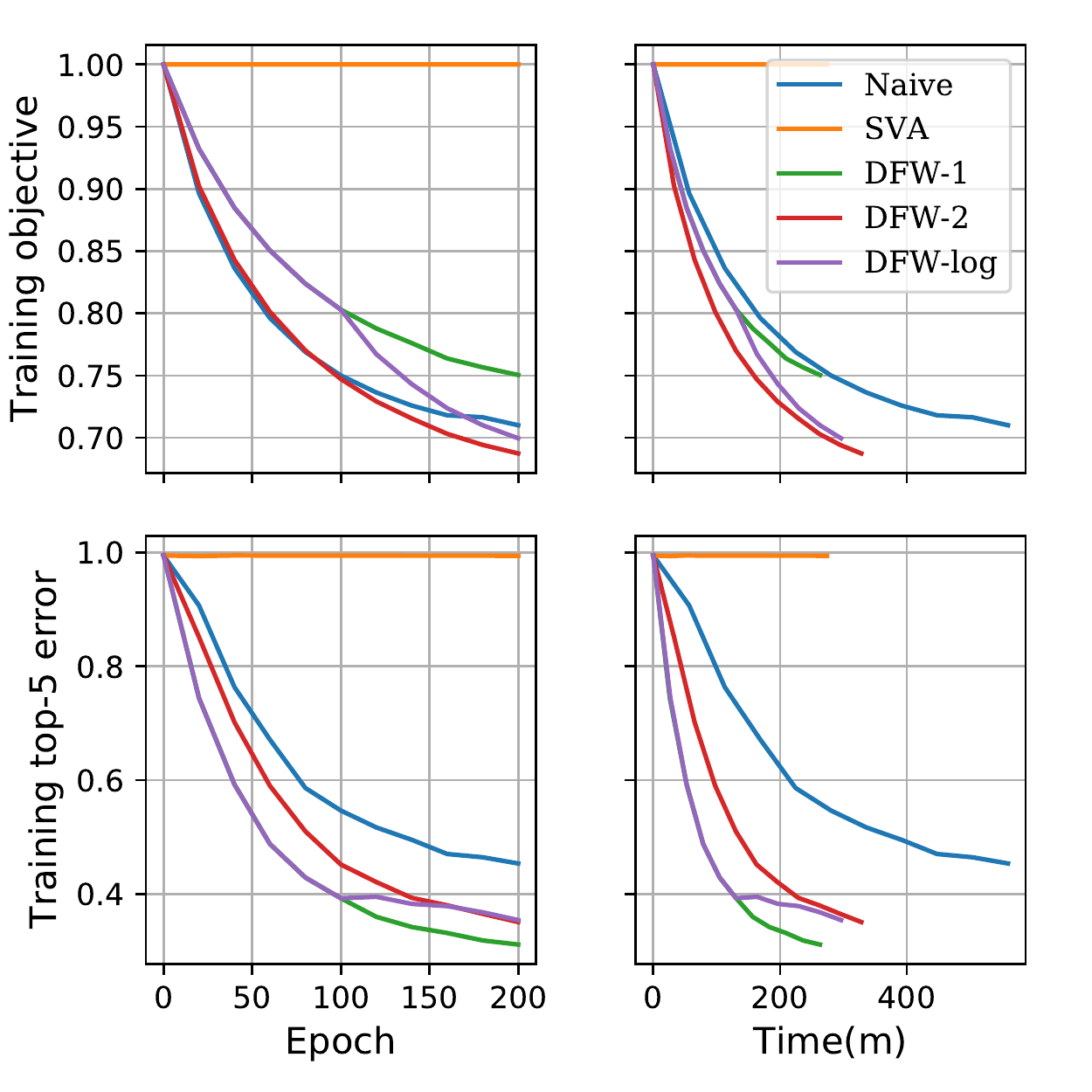} 
\includegraphics[width=.47\textwidth]{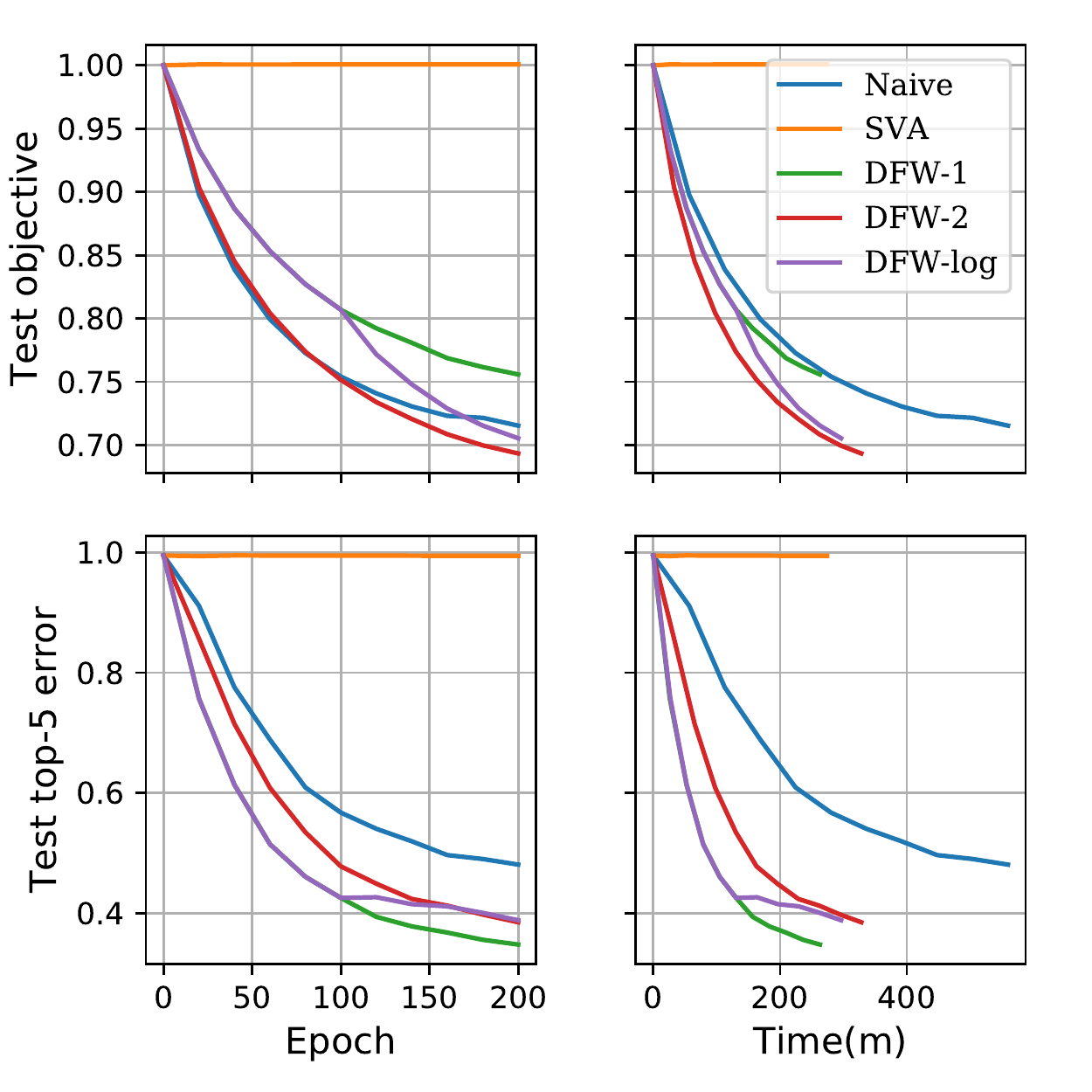} 
\caption{Results for multinomial logistic regression (ImageNet dataset). The error stands for the top-5 misclassification rate.}
\label{fig:imagenet}
\end{figure}

\begin{figure}[t]
\centering
\includegraphics[width=.37\textwidth]{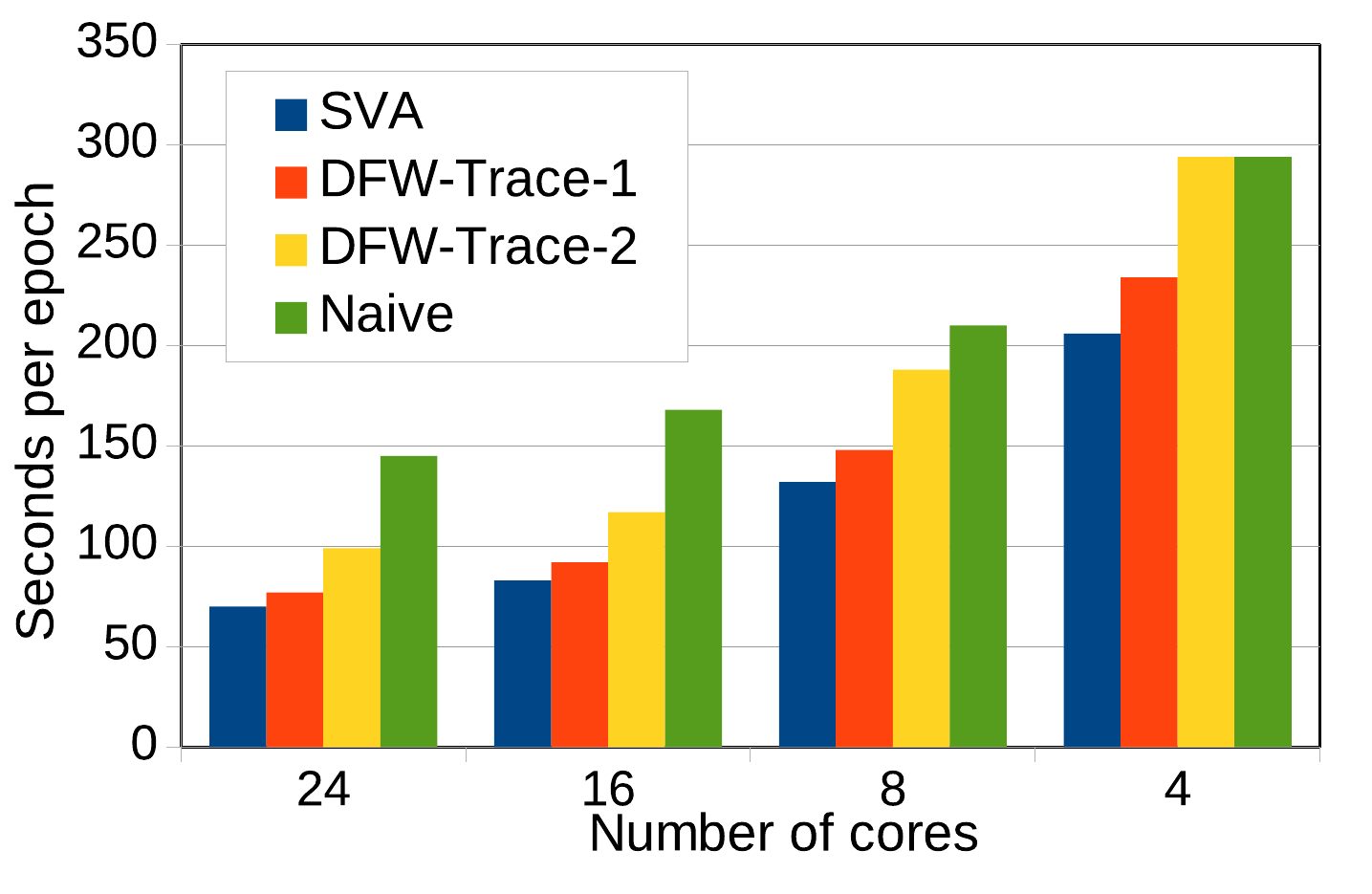}
\includegraphics[width=.37\textwidth]{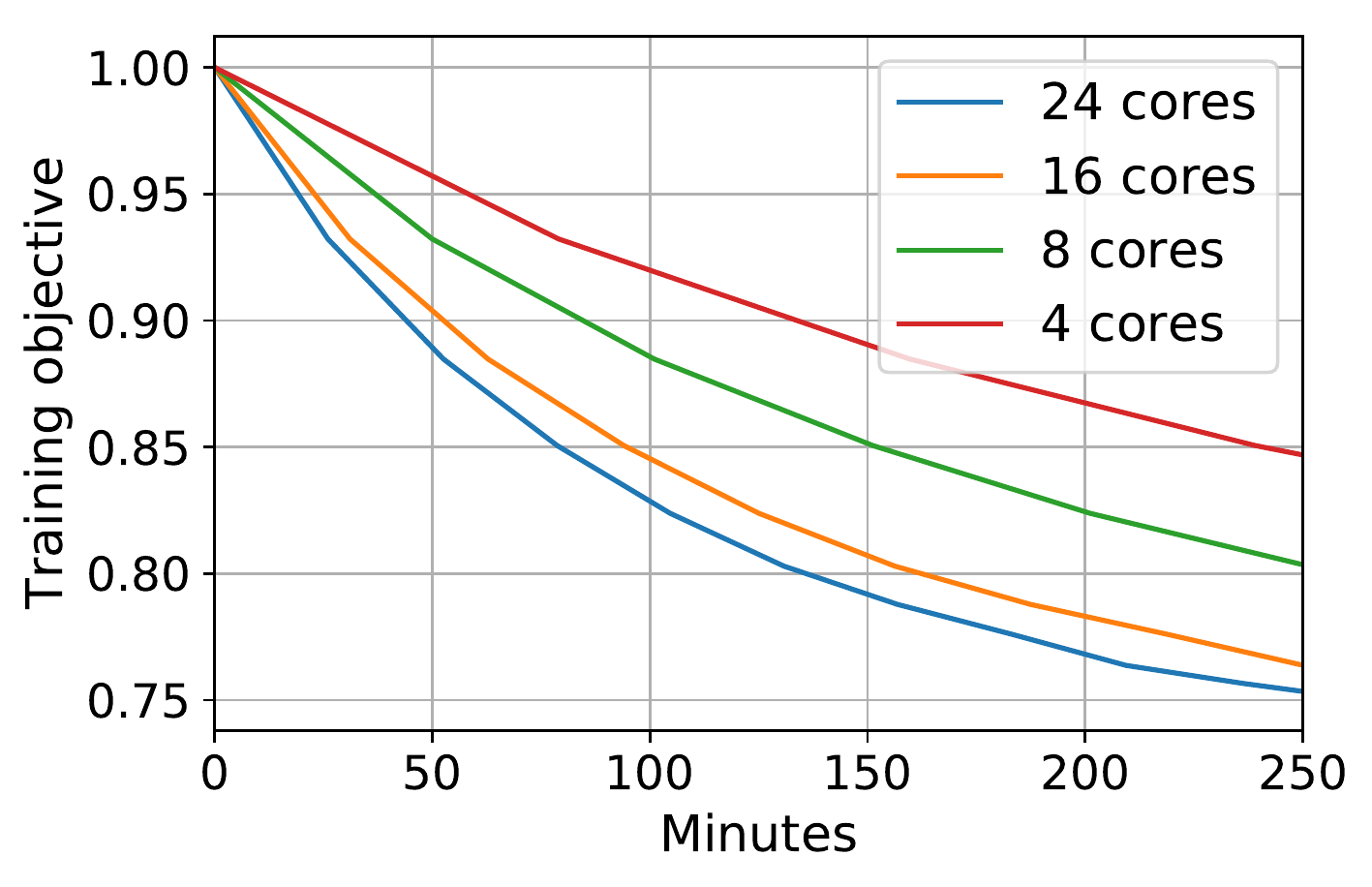}
\caption{Speed-ups with respect to the number of cores (ImageNet dataset). Left: time per epoch. Right: objective value with respect to runtime for \dfw-1.}
\label{fig:speed}
\end{figure}

\noindent\textbf{Multinomial logistic regression.}
Here, all algorithms use a fixed step size as there is no closed-form line search. As we observed empirically that this task requires a larger number of FW iterations to converge, we set $K(t)=\lfloor1+0.5\log(t)\rfloor$ for \dfw-log so that the number of power iterations does not exceed 2 as in the previous experiment.
Figure~\ref{fig:mlr} shows the results on the synthetic dataset for several values of$\mu$ (the upper bound on the trace norm). They are consistent with those obtained for multi-task least square. In particular, \sva achieves converges to a suboptimal solution, while \naive converges fast in terms of epochs but its runtime is larger than \dfw. \dfw-2 and \dfw-log perform well across all values of $\mu$: this confirms that very few power iterations are sufficient to ensure good convergence. For more constrained problems ($\mu=10$), the error does not align very well with the objective function and hence optimizing the subproblems to lower accuracy with \dfw-1 works best.

We now turn to the ImageNet dataset. The results for $\mu=30$ with $24$ cores are shown on Figure~\ref{fig:imagenet}.\footnote{The relative performance of the methods is the same for other values of $\mu$. We omit these detailed results due to the lack of space.}
Again, the \dfw variants clearly outperform \naive and \sva. While \dfw-2 and \dfw-log reduce the objective value faster than \dfw-1, the latter reduces the error slightly faster. When run until convergence, all variants converge to state-of-the-art top-5 misclassification rate with these features (around $0.13$, on par with the pre-trained deep neural net provided by Keras).

We conclude these experiments by investigating the speed-ups obtained when varying the number of cores on the ImageNet dataset. As seen on the left panel of Figure~\ref{fig:speed}, the time per epoch nicely decreases with the number of cores (with diminishing returns, as expected in distributed computing). The right panel of Figure~\ref{fig:speed} illustrates this effect on the convergence speed for \dfw-1.

%% file: conclu.tex

\section{Conclusion}
\label{sec:conclu}

In this work, we introduced a distributed Frank-Wolfe algorithm for learning high-dimensional low-rank matrices from large-scale datasets. Our \dfw algorithm is communication-efficient, enjoys provable convergence rates and can be efficiently implemented in map-reduce operations. We implemented \dfw as a Python toolbox relying on the Apache Spark distributed programming framework, and showed that it performs well on synthetic and real datasets.

In future work, we plan to investigate several directions. First, we would like to study whether faster theoretical convergence can be achieved under additional assumptions. Second, we wonder whether our algorithm can be deployed in GPUs and be used in neural networks with back-propagated gradients. Finally, we hope to explore how to best combine the ideas of distributed and stochastic Frank-Wolfe algorithms.

%% file: appendix.tex

\section{Proof of Theorem~\ref{thm:conv-exp}}
\label{sec:proof1}

Notice that our distributed version of the power method used in \dfw (Algorithm~\ref{alg:dfw}, lines~\ref{line:pmstart}--\ref{line:pmend}) exactly corresponds to the serial power method applied to the full gradient $\nabla F(W^t)$. Hence \dfw performs the same steps as a centralized Frank-Wolfe algorithm that would use the power method to approximately solve the subproblems. We will thus abstract away the details related to the distributed setting (e.g., how the data is split, how parallel computation is organized): our analysis consists in characterizing the approximation error incurred by the power method and showing that this error is small enough to ensure that the Frank-Wolfe algorithm converges in expectation.


We start by establishing that if the linear subproblem is approximately solved in expectation (to sufficient accuracy), then the standard Frank-Wolfe algorithm converges sublinearly in expectation (in the sense of Definition~\ref{def:conv}).

\begin{lemma}
\label{lem:conv-exp}
Let $\delta \ge 0$ be an accuracy parameter. If at each step $t\geq 0$, the linear subproblem is approximately solved in expectation, i.e. we find a random variable $\hat{S}$ such that
\begin{equation}
\label{eq:sub-exp}
\langle \mathbb{E}[\hat{S}|W^t], \nabla F(W^t) \rangle \le \min_{S\in\mathcal{D}} \langle S, \nabla F(W^t) \rangle + \tfrac{1}{2} \delta \gamma^t C_F,
\end{equation}
then the Frank-Wolfe algorithm converges sublinearly in expectation.
\end{lemma}
\begin{proof}
At any step $t$, given $W^t$ we set $W^{t+1} = W^t + \gamma^t (\hat{S} - W^t)$ with arbitrary step size $\gamma^t \in [0,1]$. From the definition of the curvature constant $C_F$ \citep{jaggi2013revisiting}:
\[
F(W^{t+1}) \le F(W^t) + \gamma^t \langle \hat{S}-W^t, \nabla F(W^t) \rangle + \tfrac{(\gamma^t)^2}{2} C_F .
\]
We can now take conditional expectation on both sides and use \eqref{eq:sub-exp} to get
\begin{align*}
\mathbb{E}[F(W^{t+1})|W^t] &\le F(W^t) + \gamma^t \langle \mathbb{E}[\hat{S}|W^t]-W^t, \nabla F(W^t) \rangle + \tfrac{(\gamma^t)^2}{2} C_F \\
& \le F(W^t) + \gamma^t \left(\min_{S\in\mathcal{D}} \langle S - W^t, \nabla F(W^t) \rangle  \right) + \tfrac{(\gamma^t)^2}{2} C_F (1 + \delta) \\
& \le F(W^t) - \gamma^t G(W^t) + (\gamma^t)^2 C,
\end{align*}
where we denote $G(W) := \max_{S\in\mathcal{D}} \langle W-S, \nabla F(W) \rangle$ and $C:=\tfrac{C_F}{2}(1 + \delta)$. The function $G(W)$ is known as the \emph{duality gap} and satisfies $F(W)-F(W^*) \le G(W)$ --- see \cite{jaggi2013revisiting} for details.
Denoting $H(W):=F(W)-F(W^*)$, we have
\begin{align*}
\mathbb{E}[H(W^{t+1})|W^t] & \le H(W^t) - \gamma^t G(W^t) + (\gamma^t)^2 C \\
& \le H(W^t) - \gamma^t H(W^t) + (\gamma^t)^2 C \\
& = (1-\gamma^t)H(W^t) + (\gamma^t)^2 C,
\end{align*}
where we use the duality $H(x) \le G(x)$.

We shall use induction over $t$ to prove the sublinear convergence in expectation \eqref{eq:con-exp}, i.e., we want to show that
\[
\mathbb{E}[H(W^t)] \le \tfrac{4C}{t+2}, \quad \text{for } t=1,2,...
\]
We prove this for the default step size $\gamma^t=\frac{2}{t+2}$ (we can easily prove the same thing for the line search variant, as the resulting iterates always achieve a lower objective than with the default step size).
For $t=1$, we have $\gamma^0 = \tfrac{2}{0+2}=1$. For any $W \in \mathcal{D}$, we have $H(W) \le \tfrac{C_F}{2} < C < \tfrac{4}{3}C$. This proves the case of $t=1$.
Consider now $t\geq 2$, then
\begin{align*}
\mathbb{E} [H(W^{t+1})] = \mathbb{E}[\mathbb{E} [H(W^{t+1})|W^t]] & \le (1-\gamma^t) \mathbb{E} [H(W^t)] + (\gamma^t)^2 C \\
& \le \left(1- \tfrac{2}{t+2} \right) \tfrac{4C}{t+2} + \left( \tfrac{2}{t+2} \right)^2 C.
\end{align*}
Simply rearranging the terms gives
\[
\mathbb{E} [H(W^{t+1})] 
\le \tfrac{4(t+1)C}{(t+2)^2} 
< \tfrac{4(t+1)C}{(t+1)(t+3)} 
= \tfrac{4C}{t+3} .
\]
This concludes the proof. 
\end{proof}

Based on Lemma~\ref{lem:conv-exp}, in order to prove Theorem~\ref{thm:conv-exp} we need to quantify the number of power method iterations needed to achieve the desired accuracy \eqref{eq:sub-exp} for the linear subproblems. We will rely on some results from \citet[][Theorem~3.1 therein]{kuczynski1992estimating}, which we recall in the lemma below.

\begin{lemma}[\citealp{kuczynski1992estimating}]
\label{lem:eigen}
Let $A\in\mathbb{R}^{m \times m}$ be any symmetric and positive definite matrix, and $b$ be a random vector chosen uniformly on the unit sphere (with $P$ the corresponding probability measure). Denote by $\lambda_1$ the largest eigenvalue of $A$ and by $\xi=\xi(A, b, K)$ the estimate given by $K$ power iterations. We define its average relative error $e(\xi)$ as
\begin{equation*}
e(\xi) := \int_{\lVert b \rVert = 1} \left| \frac{\xi-\lambda_1}{\lambda_1} \right| P (db).
\end{equation*}
Then for any $K\ge 2$ and $m \ge 8$, regardless of $A$, we have
\begin{equation*}
e(\xi) \le \alpha(m) \dfrac{\ln m}{K-1},
\end{equation*}
where $\pi^{-1/2} \le \alpha(m) \le 0.871$ and, for large $m$, $\alpha(m) \approx \pi^{-1/2} \approx 0.564$. 

Moreover, if $\lambda$ has multiplicity 1, denoting the second largest eigenvalue by $\lambda_2$, then there exists a constant $\widetilde{K}$, so that for any $K > \widetilde{K}$, we have
\begin{equation*}
e(\xi) \le m \left(\frac{\lambda_2}{\lambda_1}\right)^{K-1} .
\end{equation*}
\end{lemma}

We introduce a last technical lemma.

\begin{lemma}
\label{lem:lip}
If a differentiable function $F$ is $L$-Lipschitz continuous w.r.t. the trace norm, then for any matrix $W$, all singular values of $\nabla F(W)$ are smaller than $L$.
\end{lemma}
\begin{proof}
For any matrix $W$, the definition of $L$-Lipschitzness implies that
\[
\sup_{\Delta W \neq 0} \dfrac{|F(W+\Delta W) - F(W)|}{\lVert \Delta W \rVert_*} \le L.
\]
According to the mean value theorem, there exists a matrix $X$ between $W$ and $W+\Delta W$ such that
\[
\sup_{\Delta W \neq 0} \left< \nabla F(X), \frac{\Delta W}{\lVert \Delta W \rVert_*} \right> \le L.
\]
Denote the largest singular value of $W$ by $\sigma_1(W)$. Since the spectral norm is the dual norm of the trace norm, we have 
$\sigma_1(\nabla F(X)) \le L$.
Letting $\Delta W \rightarrow 0$, we get
$\sigma_1(\nabla F(W)) \le L$.
\end{proof}

Based on the above intermediary results, we can now prove Theorem~\ref{thm:conv-exp}.
For any $t\geq 0$, denote $A^t:=\nabla F(W^t)$. The largest eigenvalue of ${A^t}^\top A^t$ is the square of the largest singular value of $A^t$, denoted as $\sigma^t_1$. We estimate $(\sigma_1^t)^2$ as $v_{K(t)}^\top {A^t}^\top A^t v_{K(t)}$, where $v_{K(t)}$ is the normalized unit vector after $K(t)$ power iterations.
We also denote $u_{K(t)} := A^tv_{K(t)} / \lVert A^tv_{K(t)} \rVert $.

According to Lemma~\ref{lem:eigen}, we have
\[
\mathbb{E} \left| \dfrac{v_{K(t)}^\top  {A^t}^\top A^t v_{K(t)} - (\sigma_1^t)^2}{(\sigma_1^t)^2} \right| \le  \dfrac{\ln m}{K(t)-1}.
\]

Therefore:
\begin{eqnarray*}
\mathbb{E} \left| \dfrac{\lVert A^tv_{K(t)} \rVert}{\sigma^t_1} - 1 \right|
&\le& \mathbb{E} \left| \dfrac{\lVert A^tv_{K(t)} \rVert}{\sigma^t_1} - 1 \right| \left| \dfrac{\lVert A^tv_{K(t)}\rVert}{\sigma^t_1} + 1 \right|\\
&=& \mathbb{E} \left| \dfrac{\lVert A^tv_{K(t)} \rVert^2}{(\sigma^t_1)^2} - 1 \right| 
\le  \dfrac{\ln m}{K(t)-1}.
\end{eqnarray*}

Let $K(t)=1+\lceil\frac{\mu L (t+2) \ln m}{\delta C_F}\rceil$, we get
\[
\mathbb{E} \left| \dfrac{\lVert A^tv_{K(t)} \rVert}{\sigma^t_1} - 1 \right| 
\le \frac{1}{2} \dfrac{\delta \gamma^t C_F}{\mu L} 
\le \frac{1}{2} \dfrac{\delta \gamma^t C_F}{\mu \sigma^t_1} ,
\]
where the last inequality uses Lemma~\ref{lem:lip}.

Removing the absolute sign and the denominator, we get
\[
\mathbb{E} [\mu (\sigma^t_1 - \lVert A^tv_{K(t)} \rVert)] \le \tfrac{1}{2} \delta \gamma^t C_F .
\]

Rearranging the terms, we obtain
\begin{equation}
\label{eq:sub-proof}
-\mu\mathbb{E} \lVert A^tv_{K(t)} \rVert \le -\mu\sigma^t_1 + \tfrac{1}{2} \delta \gamma^t C_F .
\end{equation}

On the other hand, we have
\begin{equation}
\label{eq:sub-proof2}
\mathbb{E} \lVert A^tv_{K(t)} \rVert 
= \mathbb{E} \left[\dfrac{v_{K(t)}^\top {A^t}^\top A^tv_{K(t)}}{\lVert A^tv_{K(t)} \rVert}\right]
= \mathbb{E} [u_{K(t)}^\top A^tv_{K(t)}]
= \mathbb{E} \left< u_{K(t)}v_{K(t)}^\top, A^t \right> ,
\end{equation}
and
\begin{equation}
\label{eq:sub-proof3}
\mu\sigma^t_1 = \max_{\lVert S \rVert_* \le \mu} \left< S, A^t \right> .
\end{equation}

Replacing \eqref{eq:sub-proof2} and \eqref{eq:sub-proof3} into \eqref{eq:sub-proof}, we obtain \eqref{eq:sub-exp}. The first assertion of Theorem~\ref{thm:conv-exp} thus holds by application of Lemma~\ref{lem:conv-exp}. For the second assertion, the proof is nearly identical. Indeed, by replacing $\tfrac{\ln m}{K(t)-1}$ with $m \beta^{2K(t)-2}$, we get the desired result.

\section{Implementation Details for Two Tasks}
\label{sec:imp-app}

For the two tasks studied in this paper, we describe the \emph{sufficient information} maintained by workers and how to efficiently update it. Table~\ref{tab:complexity} summarizes the per-worker time and memory complexity of \dfw depending on the representation used for the sufficient information. Generally, the low-rank representation is more efficient when the number of local data points $n_j < \min(d, m)$.
\\

\noindent\textbf{Multi-task least square regression.} Recalling the multi-task regression formulation in \eqref{eq:mls}, for any worker $j$ we will denote by $X_j$ the $n_j\times d$ matrix representing the feature representation of the data points held by $j$. Similarly, we use $Y_j$ to denote the $n_j\times m$ response matrix associated with these data points.
The gradient of \eqref{eq:mls} is given by $\nabla F(W) = X^\top (XW-Y)$. At each step $t$, each worker $j$ will store $(X_j^\top Y_j,X_j^\top X_j, X_j^\top X_jW^t, W^t, \nabla F_j(W^t))$ as sufficient information. The quantities $X_j^\top Y_j$ and $X_j^\top X_j$ are fixed and precomputed. Given $W^t$, $W^{t+1}=(1-\gamma^t)W^t + \gamma^t S^t$ is efficiently obtained by rescaling $W^t$ and adding the rank-1 matrix $\gamma^t S^t$. A similar update scheme is used for $X_j^\top X_jW^t$. Assuming $W^0$ is initialized to the zero matrix, the local gradient is initialized as $\nabla F_j(W^0) = -X_j^\top Y_j$ and can be efficiently updated using the following formula:
\begin{eqnarray*}
\nabla F_j(W^{t+1}) &=& X_j^\top (X_jW^{t+1}-Y_j) = X_j^\top (X_j[(1-\gamma^t)W^t + \gamma^t S^t]-Y_j)\\
&=& (1-\gamma^t)\nabla F_j(W^{t}) + \gamma^t (X_j^\top X_j S^t - X_j^\top Y_j) .
\end{eqnarray*}

The same idea can be applied to perform line search, as the optimal step size at any step $t$ is given by the following closed-form formula:
$$\gamma^t = \frac{\left< -\nabla f(W^t), S^t-W^t \right>}{\left< X^\top X(S^t-W), S^t-W \right>}.$$
\\

\noindent\textbf{Multinomial logistic regression.} We now turn to the multi-class classification problem \eqref{eq:mlr}. As above, for a worker $j$ we denote by $X_j$ its local $n_j\times d$ feature matrix and by $Y_j\in\mathbb{R}^{n_j}$ the associated labels.
The gradient of \eqref{eq:mlr} is given by $\nabla F(W) = X^\top (P - H)$, where $P$ and $H$ are $n\times m$ matrices whose entries $(i,l)$ are $P_{il} = \frac{\exp(w_l^T x_i)}{\sum_k \exp(w_k^T x_i)}$ and $H_{il} = \mathbb{I}[y_i = l]$ respectively. The sufficient information stored by worker $j$ at each step $t$ is $(X_j, X_j^\top H_j, X_jW^t, \nabla F_j(W^t))$. $X_j^\top H_j$ is fixed and precomputed. Assuming that $W^0$ is the zero matrix, $X_jW^t$ is initialized to zero and easily updated through a low-rank update. The local gradient $\nabla F_j(W^t) = X_j^\top P_j - X_j^\top H_j$ can then be obtained by applying the softmax operator on $X_jW^t$.
Note that there is no closed-form for the line search.

\begin{table}[t]
\caption{Time and memory complexity of \dfw for the $j$-th worker on two tasks with dense vs. low-rank representations for the sufficient information.}
\label{tab:complexity}
\centering
\begin{tabular}{ccccc}
\hline\noalign{\smallskip}
 & \multicolumn{2}{c}{\textbf{Multi-task least square}} & \multicolumn{2}{c}{\textbf{Multinomial logistic regression}} \\ 
\noalign{\smallskip}
& Dense & Low-rank & Dense & Low-rank\\
\noalign{\smallskip}\hline\noalign{\smallskip}
Init. & $O(n_j(d^2+md))$ & 0 & $O(n_jd+md)$ & 0 \\ 
Power iter. & $O(md)$ & $O(n_j(d+m))$ & $O(md)$ & $O(n_j(d+m))$ \\
Update & $O(d^2+md)$ & $O(n_j(d+m))$ & $O(n_jmd)$ & $O(n_j(d+m))$ \\ 
Line search & $O(d^2+md)$ & $O(n_jm)$ & --- & --- \\ 
\noalign{\smallskip}\hline\noalign{\smallskip}
Memory & $O(d^2+md)$ & $O(n_j(d+m))$ & $O(n_j(d+m)+md)$ & $O(n_j(d+m))$ \\ 
\noalign{\smallskip}\hline
\end{tabular}
\end{table}

%% file: main.bbl
\begin{thebibliography}{52}
\providecommand{\natexlab}[1]{#1}
\providecommand{\url}[1]{{#1}}
\providecommand{\urlprefix}{URL }
\expandafter\ifx\csname urlstyle\endcsname\relax
  \providecommand{\doi}[1]{DOI~\discretionary{}{}{}#1}\else
  \providecommand{\doi}{DOI~\discretionary{}{}{}\begingroup
  \urlstyle{rm}\Url}\fi
\providecommand{\eprint}[2][]{\url{#2}}

\bibitem[{Amit et~al(2007)Amit, Fink, Srebro, and Ullman}]{Amit2007a}
Amit Y, Fink M, Srebro N, Ullman S (2007) Uncovering shared structures in
  multiclass classification. In: ICML

\bibitem[{Argyriou et~al(2008)Argyriou, Evgeniou, and Pontil}]{Argyriou2008a}
Argyriou A, Evgeniou T, Pontil M (2008) {C}onvex multi-task feature learning.
  {M}achine {L}earning 73(3):243--272

\bibitem[{Bach(2008)}]{bach2008consistency}
Bach FR (2008) Consistency of trace norm minimization. Journal of Machine
  Learning Research 9:1019--1048

\bibitem[{Bellet et~al(2015)Bellet, Liang, Garakani, Balcan, and
  Sha}]{Bellet2015b}
Bellet A, Liang Y, Garakani AB, Balcan MF, Sha F (2015) {A} {D}istributed
  {F}rank-{W}olfe {A}lgorithm for {C}ommunication-{E}fficient {S}parse
  {L}earning. In: SDM

\bibitem[{Bhojanapalli et~al(2016)Bhojanapalli, Neyshabur, and
  Srebro}]{bhojanapalli2016global}
Bhojanapalli S, Neyshabur B, Srebro N (2016) {Global Optimality of Local Search
  for Low Rank Matrix Recovery}. In: NIPS

\bibitem[{Bro et~al(2008)Bro, Acar, and Kolda}]{Bro2008a}
Bro R, Acar E, Kolda TG (2008) Resolving the sign ambiguity in the singular
  value decomposition. Journal of Chemometrics 22(2):135--140

\bibitem[{Cabral et~al(2013)Cabral, De~La~Torre, Costeira, and
  Bernardino}]{cabral2013unifying}
Cabral R, De~La~Torre F, Costeira JP, Bernardino A (2013) Unifying nuclear norm
  and bilinear factorization approaches for low-rank matrix decomposition. In:
  ICCV

\bibitem[{Cabral et~al(2011)Cabral, De~la Torre, Costeira, and
  Bernardino}]{cabral2011matrix}
Cabral RS, De~la Torre F, Costeira JP, Bernardino A (2011) {Matrix Completion
  for Multi-label Image Classification}. In: NIPS

\bibitem[{Cai et~al(2010)Cai, Cand{\`e}s, and Shen}]{cai2010singular}
Cai JF, Cand{\`e}s EJ, Shen Z (2010) A singular value thresholding algorithm
  for matrix completion. SIAM Journal on Optimization 20(4):1956--1982

\bibitem[{Cand{\`e}s and Recht(2009)}]{candes2009exact}
Cand{\`e}s EJ, Recht B (2009) Exact matrix completion via convex optimization.
  Foundations of Computational mathematics 9(6):717--772

\bibitem[{Cand{\`e}s and Tao(2010)}]{candes2010power}
Cand{\`e}s EJ, Tao T (2010) {The power of convex relaxation: Near-optimal
  matrix completion}. IEEE Transactions on Information Theory 56(5):2053--2080

\bibitem[{Candes et~al(2015)Candes, Eldar, Strohmer, and
  Voroninski}]{candes2015phase}
Candes EJ, Eldar YC, Strohmer T, Voroninski V (2015) Phase retrieval via matrix
  completion. SIAM Review 57(2):225--251

\bibitem[{Caruana(1997)}]{Caruana1997a}
Caruana R (1997) {M}ultitask {L}earning. {M}achine {L}earning 28(1):41--75

\bibitem[{Clarkson(2010)}]{clarkson2010coresets}
Clarkson KL (2010) {Coresets, sparse greedy approximation, and the Frank-Wolfe
  algorithm}. ACM Transactions on Algorithms 6(4):63

\bibitem[{Dean and Ghemawat(2008)}]{dean2008mapreduce}
Dean J, Ghemawat S (2008) Mapreduce: simplified data processing on large
  clusters. Communications of the ACM 51(1):107--113

\bibitem[{Deng et~al(2009)Deng, Dong, Socher, Li, Li, and Li}]{Deng2009a}
Deng J, Dong W, Socher R, Li LJ, Li K, Li FF (2009) {ImageNet: A large-scale
  hierarchical image database}. In: CVPR

\bibitem[{Frank and Wolfe(1956)}]{frank1956algorithm}
Frank M, Wolfe P (1956) An algorithm for quadratic programming. Naval research
  logistics quarterly 3(1-2):95--110

\bibitem[{Freund and Grigas(2016)}]{Freund2016a}
Freund RM, Grigas P (2016) {New analysis and results for the Frank--Wolfe
  method}. Mathematical Programming 155(1--2):199--230

\bibitem[{Garber and Hazan(2015)}]{Garber2015a}
Garber D, Hazan E (2015) {Faster Rates for the Frank-Wolfe Method over
  Strongly-Convex Sets}. In: ICML

\bibitem[{Goldberg et~al(2010)Goldberg, Recht, Xu, Nowak, and
  Zhu}]{goldberg2010transduction}
Goldberg A, Recht B, Xu J, Nowak R, Zhu X (2010) {Transduction with matrix
  completion: Three birds with one stone}. In: NIPS

\bibitem[{Gross(2011)}]{gross2011recovering}
Gross D (2011) Recovering low-rank matrices from few coefficients in any basis.
  IEEE Transactions on Information Theory 57(3):1548--1566

\bibitem[{Gross et~al(2010)Gross, Liu, Flammia, Becker, and
  Eisert}]{gross2010quantum}
Gross D, Liu YK, Flammia ST, Becker S, Eisert J (2010) Quantum state tomography
  via compressed sensing. Physical review letters 105(15):150,401

\bibitem[{Harchaoui et~al(2012)Harchaoui, Douze, Paulin, Dudik, and
  Malick}]{harchaoui2012large}
Harchaoui Z, Douze M, Paulin M, Dudik M, Malick J (2012) Large-scale image
  classification with trace-norm regularization. In: CVPR

\bibitem[{Harchaoui et~al(2015)Harchaoui, Juditsky, and
  Nemirovski}]{Harchaoui2015a}
Harchaoui Z, Juditsky A, Nemirovski A (2015) Conditional gradient algorithms
  for norm-regularized smooth convex optimization. {M}athematical {P}rogramming
  152(1--2):75--112

\bibitem[{Hazan(2008)}]{hazan2008sparse}
Hazan E (2008) Sparse approximate solutions to semidefinite programs. In: Latin
  American Symposium on Theoretical Informatics

\bibitem[{Hazan and Kale(2012)}]{Hazan2012a}
Hazan E, Kale S (2012) {Projection-free Online Learning}. In: ICML

\bibitem[{Hazan and Luo(2016)}]{Hazan2016a}
Hazan E, Luo H (2016) {Variance-Reduced and Projection-Free Stochastic
  Optimization}. In: ICML

\bibitem[{He et~al(2016)He, Zhang, Ren, and Sun}]{he2016deep}
He K, Zhang X, Ren S, Sun J (2016) Deep residual learning for image
  recognition. In: CVPR

\bibitem[{Jaggi(2013)}]{jaggi2013revisiting}
Jaggi M (2013) {Revisiting Frank-Wolfe: Projection-Free Sparse Convex
  Optimization}. In: ICML

\bibitem[{Jaggi et~al(2010)Jaggi, Sulovsk et~al}]{jaggi2010simple}
Jaggi M, Sulovsk M, et~al (2010) A simple algorithm for nuclear norm
  regularized problems. In: ICML

\bibitem[{Ji et~al(2010)Ji, Liu, Shen, and Xu}]{ji2010robust}
Ji H, Liu C, Shen Z, Xu Y (2010) Robust video denoising using low rank matrix
  completion. In: CVPR

\bibitem[{Koltchinskii et~al(2011)Koltchinskii, Lounici, and
  Tsybakov}]{koltchinskii2011nuclear}
Koltchinskii V, Lounici K, Tsybakov AB (2011) Nuclear-norm penalization and
  optimal rates for noisy low-rank matrix completion. The Annals of Statistics
  39(5):2302--2329

\bibitem[{Koren et~al(2009)Koren, Bell, Volinsky et~al}]{koren2009matrix}
Koren Y, Bell R, Volinsky C, et~al (2009) Matrix factorization techniques for
  recommender systems. Computer 42(8):30--37

\bibitem[{Kuczy{\'n}ski and Wo{\'z}niakowski(1992)}]{kuczynski1992estimating}
Kuczy{\'n}ski J, Wo{\'z}niakowski H (1992) {Estimating the largest eigenvalue
  by the power and Lanczos algorithms with a random start}. SIAM Journal on
  Matrix Analysis and Applications 13(4):1094--1122

\bibitem[{Lacoste-Julien and Jaggi(2015)}]{Lacoste-Julien2015a}
Lacoste-Julien S, Jaggi M (2015) {On the Global Linear Convergence of
  Frank-Wolfe Optimization Variants}. In: NIPS

\bibitem[{Lacoste-Julien et~al(2013)Lacoste-Julien, Jaggi, Schmidt, and
  Pletscher}]{Lacoste-Julien2013b}
Lacoste-Julien S, Jaggi M, Schmidt M, Pletscher P (2013) {B}lock-{C}oordinate
  {F}rank-{W}olfe {O}ptimization for {S}tructural {SVM}s. In: ICML

\bibitem[{Lan and Zhou(2016)}]{Lan2016a}
Lan G, Zhou Y (2016) {Conditional Gradient Sliding for Convex Optimization}.
  SIAM Journal on Optimization 26(2):1379--1409

\bibitem[{Liu and Tsang(2017)}]{liu2017approximate}
Liu Z, Tsang I (2017) {Approximate Conditional Gradient Descent on Multi-Class
  Classification}. In: AAAI

\bibitem[{Ma et~al(2011)Ma, Goldfarb, and Chen}]{ma2011fixed}
Ma S, Goldfarb D, Chen L (2011) {Fixed point and Bregman iterative methods for
  matrix rank minimization}. Mathematical Programming 128(1-2):321--353

\bibitem[{Mackey et~al(2011)Mackey, Jordan, and Talwalkar}]{mackey2011divide}
Mackey LW, Jordan MI, Talwalkar A (2011) Divide-and-conquer matrix
  factorization. In: Advances in Neural Information Processing Systems, pp
  1134--1142

\bibitem[{Moharrer and Ioannidis(2017)}]{moharrerdistributing}
Moharrer A, Ioannidis S (2017) {Distributing Frank-Wolfe via Map-Reduce}. In:
  ICDM

\bibitem[{Parikh and Boyd(2013)}]{Parikh2013a}
Parikh N, Boyd S (2013) Proximal algorithms. Foundations and Trends in
  Optimization 1(3):123--231

\bibitem[{Pong et~al(2010)Pong, Tseng, Ji, and Ye}]{pong2010trace}
Pong TK, Tseng P, Ji S, Ye J (2010) Trace norm regularization: Reformulations,
  algorithms, and multi-task learning. SIAM Journal on Optimization
  20(6):3465--3489

\bibitem[{Recht(2011)}]{recht2011simpler}
Recht B (2011) A simpler approach to matrix completion. Journal of Machine
  Learning Research 12:3413--3430

\bibitem[{Russakovsky et~al(2015)Russakovsky, Deng, Su, Krause, Satheesh, Ma,
  Huang, Karpathy, Khosla, Bernstein, Berg, and Fei-Fei}]{ILSVRC15}
Russakovsky O, Deng J, Su H, Krause J, Satheesh S, Ma S, Huang Z, Karpathy A,
  Khosla A, Bernstein M, Berg AC, Fei-Fei L (2015) {ImageNet Large Scale Visual
  Recognition Challenge}. International Journal of Computer Vision
  115(3):211--252

\bibitem[{Sturm(1999)}]{sturm1999using}
Sturm JF (1999) {Using SeDuMi 1.02, a MATLAB toolbox for optimization over
  symmetric cones}. Optimization methods and software 11(1-4):625--653

\bibitem[{Toh et~al(1999)Toh, Todd, and T{\"u}t{\"u}nc{\"u}}]{toh1999sdpt3}
Toh KC, Todd MJ, T{\"u}t{\"u}nc{\"u} RH (1999) {SDPT3—a MATLAB software
  package for semidefinite programming, version 1.3}. Optimization methods and
  software 11(1-4):545--581

\bibitem[{Tran et~al(2015)Tran, Peel, and Skhiri}]{Tran2015a}
Tran NL, Peel T, Skhiri S (2015) Distributed frank-wolfe under pipelined stale
  synchronous parallelism. In: IEEE Big Data

\bibitem[{Wai et~al(2017{\natexlab{a}})Wai, Lafond, Scaglione, and
  Moulines}]{wai2017decentralized}
Wai HT, Lafond J, Scaglione A, Moulines E (2017{\natexlab{a}}) {Decentralized
  Frank-Wolfe Algorithm for Convex and Non-convex Problems}. IEEE Transactions
  on Automatic Control 62:5522--5537

\bibitem[{Wai et~al(2017{\natexlab{b}})Wai, Scaglione, Lafond, and
  Moulines}]{wai2017fast}
Wai HT, Scaglione A, Lafond J, Moulines E (2017{\natexlab{b}}) Fast and privacy
  preserving distributed low-rank regression. In: ICASSP

\bibitem[{Wang et~al(2016)Wang, Sadhanala, Dai, Neiswanger, Sra, and
  Xing}]{wang2016parallel}
Wang YX, Sadhanala V, Dai W, Neiswanger W, Sra S, Xing E (2016) {Parallel and
  distributed block-coordinate Frank-Wolfe algorithms}. In: ICML

\bibitem[{Zaharia et~al(2010)Zaharia, Chowdhury, Franklin, Shenker, and
  Stoica}]{Zaharia:2010:SCC:1863103.1863113}
Zaharia M, Chowdhury M, Franklin MJ, Shenker S, Stoica I (2010) {Spark: Cluster
  Computing with Working Sets}. In: HotCloud

\end{thebibliography}
